\newif\iflong
\longtrue

\iflong
\newcommand{\refappendix}[1]{\Cref{#1}}
\else
\newcommand{\refappendix}[1]{the SI}
\fi

\documentclass[%
 reprint,
superscriptaddress,
 amsmath,amssymb,
 aps,
 prx,
]{revtex4-2}

\usepackage{graphicx}%
\usepackage{dcolumn}%
\usepackage{bm}%
\usepackage[version=3]{mhchem} %
\usepackage{xcolor}
\usepackage{amsmath}
\usepackage{amsthm}
\usepackage[caption=false]{subfig}
\usepackage[ruled,vlined,linesnumbered]{algorithm2e}

\providecommand{\DontPrintSemicolon}{\dontprintsemicolon}
\usepackage{tikz}
\usepackage{xcolor}
\usepackage[export]{adjustbox}
\usepackage{hyperref}
\usepackage{cleveref}

\begin{document}

\preprint{APS/123-QED}

\title{Quadratic Scaling Bosonic Path Integral Molecular Dynamics}
\author{Yotam M. Y. Feldman}%
\affiliation
{School of Chemistry, Tel Aviv University, Tel Aviv 6997801, Israel.}
\affiliation{The Ratner Center for Single Molecule Science, Tel Aviv University, Tel Aviv 6997801, Israel.}
\author{Barak Hirshberg}%
\email{hirshb@tauex.tau.ac.il}%
\affiliation{School of Chemistry, Tel Aviv University, Tel Aviv 6997801, Israel.}%
\affiliation{The Ratner Center for Single Molecule Science, Tel Aviv University, Tel Aviv 6997801, Israel.}
\affiliation{The Center for Computational Molecular and Materials Science, Tel Aviv University, Tel Aviv 6997801, Israel.}%

\date{\today}%
\newcommand{\TODO}[1]{\textcolor{red}{TODO: {#1}}}
\newcommand{\yotam}[1]{\textcolor{blue}{\bf {#1}}}
\newcommand{\notice}[1]{\textcolor{red}{\bf {#1}}}
\newcommand{\barak}[1]{\textcolor{purple}{\bf {#1}}}
\newcommand{\yotamsmall}[1]{{\footnotesize\color{magenta}[{\bf Yotam}: #1]}}
\newcommand{\alternative}[1]{\textcolor{magenta}{#1}}

\newcommand{\set}[1]{\{{#1}\}}
\newcommand{\card}[1]{\left|{#1}\right|}
\newcommand{\sumcond}[2]{\substack{{#1}, \\ {#2}}}
\newcommand{\bigO}{\mathcal{O}}

\newcommand{\fact}[1]{{#1}!}
\newcommand{\pfact}[1]{\fact{\left(#1\right)}}
\newcommand{\Symset}[1]{\mathcal{S}({#1})}
\newcommand{\Symfromto}[2]{\mathcal{S}[{#1},{#2}]}
\newcommand{\SymN}[1]{\mathcal{S}[1,{#1}]}
\newcommand{\cycleof}[2]{c_{#1}({#2})}
\newcommand{\nextof}[3]{{#3}({#1})={#2}}
\newcommand{\prevparticle}[1]{{#1}^{-}}
\newcommand{\nextparticle}[1]{{#1}^{+}}
\newcommand{\particleup}{v}
\newcommand{\particledown}{u}
\newcommand{\cyclenotate}[1]{({#1})}

\newcommand{\pos}{{\bf R}}
\newcommand{\posbead}[2]{{\bf r}_{#1}^{#2}}
\newcommand{\beadpos}[2]{\posbead{#1}{#2}}
\newcommand{\mass}{m}

\newcommand{\springconstant}{\omega_P}
\newcommand{\springforceprefix}{\frac{1}{2} \mass \springconstant^2}
\newcommand{\rdiffsquared}[4]{\left(\beadpos{#1}{#2} - \beadpos{#3}{#4}\right)^2}
\newcommand{\interparticleforce}[1]{\springforceprefix \sum_{j=1}^{P-1}{\rdiffsquared{#1}{j+1}{#1}{j}}}

\newcommand{\repsym}{G}
\newcommand{\rep}[1]{\repsym[{#1}]}

\newcommand{\boltzmann}[1]{e^{-\beta {#1}}}

\newcommand{\Vfromto}[2]{V^{[{#1},{#2}]}}
\newcommand{\Vfrom}[1]{\Vfromto{#1}{N}}
\newcommand{\Vto}[1]{\Vfromto{1}{#1}}
\newcommand{\Vall}{\Vfromto{1}{N}}

\newcommand{\originalpotential}[1]{\Vto{#1}}
\newcommand{\originalpotentialnotation}[1]{V_B^{({#1})}}

\newcommand{\Eperm}[1]{E^{#1}}
\newcommand{\Efromto}[2]{E^{[{#1},{#2}]}}
\newcommand{\Enk}[2]{\Efromto{{#1}-{#2}+1}{{#1}}}
\newcommand{\Einterior}[1]{E_{\textit{int}}^{({#1})}}
\newcommand{\permutationthree}[3]{\left(\begin{smallmatrix} 
1 & 2 & 3\\
{#1} & {#2} & {#3}
\end{smallmatrix}\right)}

\newcommand{\snip}[2]{\rep{#1} \textit{ snips at } {#2}}
\newcommand{\fullsnip}[3]{\rep{#1} \textit{ snip from } {#2} \textit{ to } {#3}}
\newcommand{\project}[2]{{#1}_{| {#2}}}

\newcommand{\beadderive}[3]{\beadderiveexplicit{#1}{#2}{#3}}
\newcommand{\beadderiveexplicit}[3]{\nabla_{\beadpos{#1}{#2}} {#3}}
\newcommand{\beadforce}[3]{- \beadderive{#1}{#2}{#3}}

\newcommand{\Prrep}{\Prerepperm{\sigma}}
\newcommand{\Prerepperm}[1]{\Pr(\rep{#1})}
\newcommand{\Prrepnext}[2]{\Pr\left[\nextof{#1}{#2}{\rep{\sigma}}\right]} 
\begin{abstract}
Bosonic exchange symmetry leads to fascinating quantum phenomena, from exciton condensation in quantum materials to the superfluidity of liquid \ce{^{4}He}. Unfortunately, path integral molecular dynamics (PIMD) simulations of bosons are computationally prohibitive beyond $\mathord{\sim} 100$ particles, due to a cubic scaling with the system size. We present an algorithm that reduces the complexity from cubic to quadratic, allowing the first simulations of thousands of bosons using PIMD.
Our method is orders of magnitude faster, with a speedup that scales linearly with the number of particles and the number of imaginary time slices (beads).
Simulations that would have otherwise taken decades can now be done in days. 
In practice, the new algorithm eliminates most of the added computational cost of including bosonic exchange effects, making them almost as accessible as PIMD simulations of distinguishable particles.
\end{abstract}

\maketitle

\section{Introduction}
Path integral molecular dynamics (PIMD) simulations are a powerful computational tool for studying quantum condensed phases at thermal equilibrium~\cite{doi:10.1063/1.446740,doi:10.1063/1.465188}. They are also the basis of several prominent methods for approximating quantum transport properties~\cite{Althorpe2021}.
For distinguishable particles, PIMD samples the partition function of the quantum system through classical molecular dynamics (MD) simulations of ``ring polymers'', formed by connecting $P$ copies of each particle (``beads'') through harmonic springs~\cite{doi:10.1063/1.441588}.
State-of-the-art PIMD simulations can infer properties of quantum solids and liquids with up to $\mathord{\sim} 1000$ particles~\cite{doi:10.1063/5.0042572,https://doi.org/10.1002/adts.202000241,C6CP05968F,doi:10.1021/acs.jpclett.7b00391,doi:10.1063/1.3167790,doi:10.1021/acs.jpclett.5b01899,Cheng2016,Ceriotti2013,Cheng2021,Riera2019,Li2022,Gaiduk2018,Eltareb2022,doi:10.1021/acs.jctc.2c01233}, employing empirical force fields or interaction potentials that are based on density functional theory~\cite{Markland2018,GU202229,https://doi.org/10.1002/adts.202000258,doi:10.1021/acs.jpclett.6b00777}. 

Exchange symmetry between indistinguishable particles is central to quantum phenomena such as Bose-Einstein condensation and superfluidity. %
However, enforcing it poses a formidable challenge to PIMD.
It requires summing over all permutations of the identical particles in the quantum partition function, leading to a different set of ring polymers for each permutation where the rings of exchanged particles are connected~\cite{ceperley1995path}. 
As a result, a straightforward implementation of PIMD for indistinguishable particles scales exponentially, $\bigO(P\fact{N})$, where $N$ is the number of particles.
For that reason,
while efficient sampling of permutations is possible in Path Integral Monte Carlo (PIMC)~\cite{ceperley1995path,PhysRevLett.56.351,PhysRevLett.96.070601,PhysRevE.74.036701,doi:10.1063/5.0008309,doi:10.1063/1.1638997,MIURA1999115,DelMaestro2010,DelMaestro2011,adrian_del_maestro_2022_7271914}, exchange effects were not included in PIMD simulations until recently, except for very few particles~\cite{doi:10.1063/1.481652,PhysRevLett.121.140602,doi:10.1063/1.479789}.

In previous work, we showed~\cite{hirshberg2019path} that bosonic PIMD simulations can be performed with only \emph{cubic} scaling, $\bigO(PN^3)$.
Based on a recurrence relation for the ring polymer potential and forces, the algorithm allowed the first PIMD simulations of moderately large bosonic systems; the largest application to date consists of $N=128$ atoms~\cite{PhysRevLett.128.045301}. %
Fermionic systems whose sign problem is not too severe were also simulated using a related approach~\cite{doi:10.1063/5.0008720,doi:10.1063/5.0106067,doi:10.1063/5.0030760,doi:10.1063/5.0093472,doi:10.1063/5.0102460,PhysRevE.106.025309}.
Nevertheless, even for bosonic systems, the cubic scaling hindered applications to systems with more than $\mathord{\sim} 100$ particles.

In this paper, we present an algorithm that reduces the scaling of bosonic PIMD \emph{from cubic to quadratic}, $\bigO(N^2 + PN)$.
Our new algorithm generates the exact same potential, forces, and trajectories as the previous algorithm, only more efficiently.
The outcome is that bosonic PIMD simulations are $\bigO(PN)$ times faster than before. %
This significant improvement allows us to perform simulations of thousands of bosonic particles, which were previously beyond reach using PIMD. 
For example, a simulation of $N=1372$ \ce{^{4}He} atoms in the superfluid liquid phase takes $\mathord{\sim} 7$ days, but would require over 14 years with the previous algorithm.
We therefore dramatically reduce the added computational overhead of bosonic exchange in PIMD simulations in comparison to distinguishable particles.

In~\Cref{sec:theory}, we first describe the theory behind the original bosonic PIMD algorithm, and then %
explain how to reduce the scaling from cubic to quadratic in the evaluation of the ring polymer potential and forces. %
In~\Cref{sec:empirical}, we verify the improved scaling of the algorithm numerically,  and %
present the observed speedup thanks to our method in simulations of $>1000$ of ultracold trapped atoms and superfluid liquid \ce{^{4}He} atoms.
\Cref{sec:conclusions} presents our conclusions while \Cref{sec:methods} provides additional computational details.

\section{Results and Discussion}
\subsection{Theory}
\label{sec:theory}
This section explains the new algorithm and how it achieves the reduced scaling.
We start by reviewing bosonic PIMD, and the original, cubic-scaling algorithm.

\subsubsection{Background: cubic scaling bosonic PIMD}
We consider a standard quantum Hamiltonian of $N$ bosons of mass $\mass$, interacting through the potential $\hat{U}$, which is given by
$
\hat{H} = \frac{1}{2\mass} \sum_{\ell=1}^N {\bf \hat{p}}^2_\ell + \hat{U}({\bf r}_1,...,{\bf r}_N).
\label{eq:H}
$
The canonical partition function at inverse temperature $\beta = (k_BT)^{-1}$ is $\mathcal{Z} = \mathop{\mathrm{Tr}}{[ e^{-\beta \hat{H}}]}$,
where the trace is taken over a properly-symmetrized position basis.
The resulting path integral expression~\cite{tuckerman2010statistical} for this partition function is
\begin{equation}
\label{eq:int-of-sum-all-permutations}
    \mathcal{Z}_P(N, \beta) \propto \int{
        d\pos_1 \ldots d\pos_N \, 
        \frac{1}{\fact{N}}
        \sum_{\sigma}{\boltzmann{\left(\Eperm{\sigma} + \bar{U}\right)}}
    }.
\end{equation}

\Cref{eq:int-of-sum-all-permutations} is the
isomorphism~\cite{doi:10.1063/1.441588} between the partition function of bosonic quantum particles and that of a system of classical ring polymers, where $\pos_\ell = \beadpos{\ell}{1},\ldots,\beadpos{\ell}{P}$ are the positions of the $P$ polymer beads that correspond to particle $\ell$.
It is exact in the limit $P \to \infty$, 
and should be interpreted as follows:
1) The summation is over all the permutations $\sigma$ of the $N$ particles.
2) Beads of the same index associated with different particles interact through a scaled potential $\bar{U} = \frac{1}{P}\sum_{j=1}^{P}{U(\beadpos{1}{j},\ldots,\beadpos{N}{j})}$, which is invariant under particle permutations. We do not consider the potential $\bar{U}$ further since it is the same for bosonic and distinguishable particles.
3) $\Eperm{\sigma} = \Eperm{\sigma}(\pos_{1},...,\pos_N)$ is the spring energy of the ring polymer configuration that corresponds to $\sigma$, as we now explain.

Recall that a \emph{permutation} $\sigma$ maps each element $\ell$ in $1,\ldots,N$ to an element $\sigma(\ell)$ in $1,\ldots,N$ without repetition, i.e., $\sigma(\ell) \neq \sigma(\ell')$ for $\ell \neq \ell'$.
In $\Eperm{\sigma}$, the ring polymers of different particles are joined together by harmonic springs of frequency $\springconstant = \frac{\sqrt{P}}{\beta \hbar}$ according to $\sigma$, connecting the last bead of particle $\ell$ to the first bead of particle $\sigma(\ell)$:
\begin{equation}
\label{eq:eperm}
    \Eperm{\sigma} = \springforceprefix \sum_{\ell=1}^{N}{ \sum_{j=1}^{P}{\rdiffsquared{\ell}{j}{\ell}{j+1}}}, \mbox{ where $\beadpos{\ell}{P+1} = \beadpos{\sigma(\ell)}{1}$.}
\end{equation}
\Cref{fig:all-permutations} depicts all the permutations of $N=3$ particles, and the ring polymer configurations that correspond to them.
\begin{figure*}[t]
\center
\begin{minipage}{0.9\textwidth}
\center
\captionsetup[subfloat]{labelformat=empty}
\begin{minipage}{.33\textwidth}
\center
  \subfloat[$\permutationthree{1}{2}{3} = \cyclenotate{1}\cyclenotate{2}\cyclenotate{3}$]{\includegraphics[trim={3.5in 1in 3.5in 1in}, clip, width=0.8\textwidth, cfbox=blue, frame]{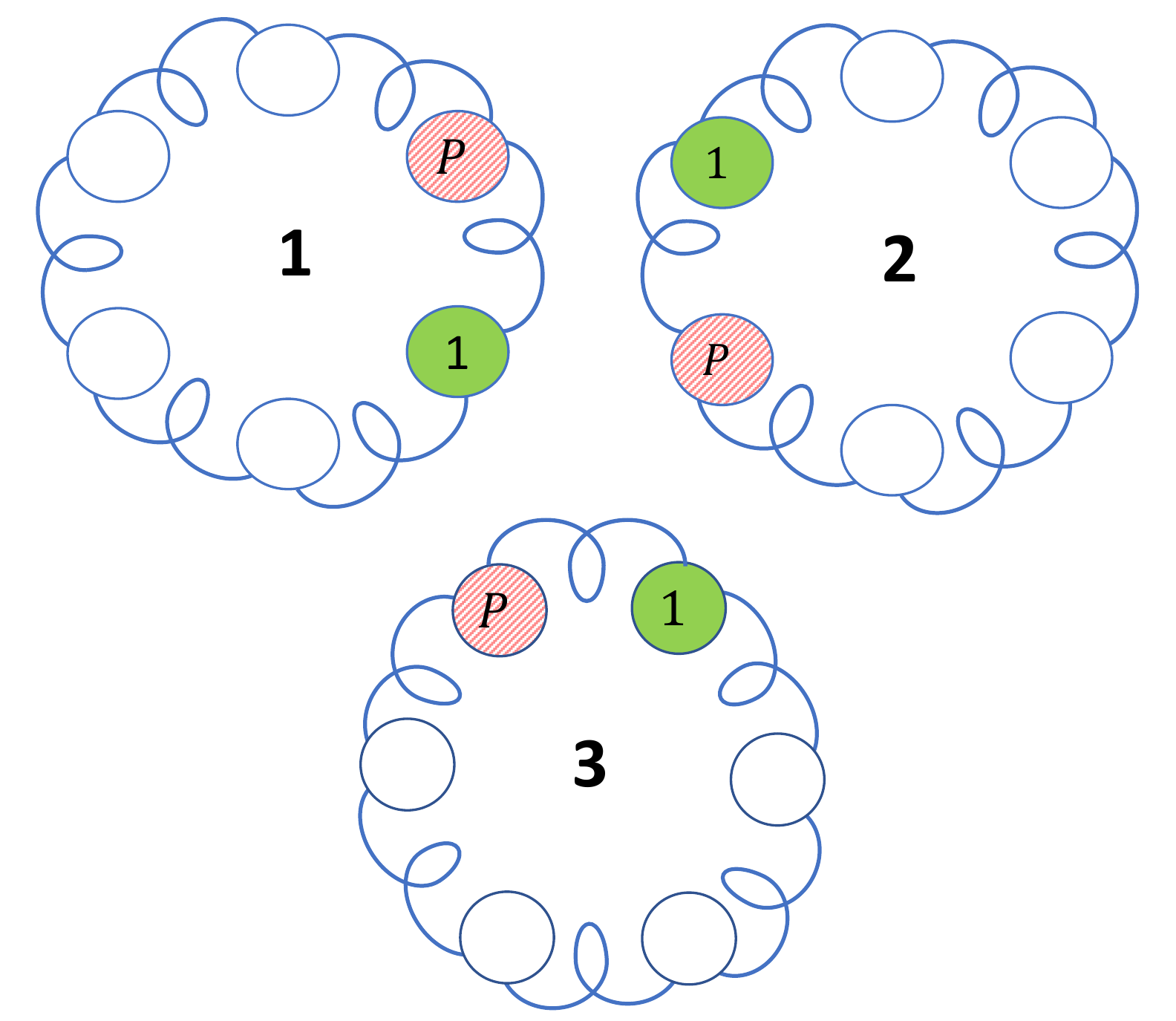}
    }
\end{minipage}%
\begin{minipage}{.33\textwidth}
\center
    \subfloat[$\permutationthree{2}{3}{1} = \cyclenotate{123}$]{\includegraphics[trim={3.5in 1in 3.5in 1in}, clip, width=0.8\textwidth, cfbox=blue, frame]{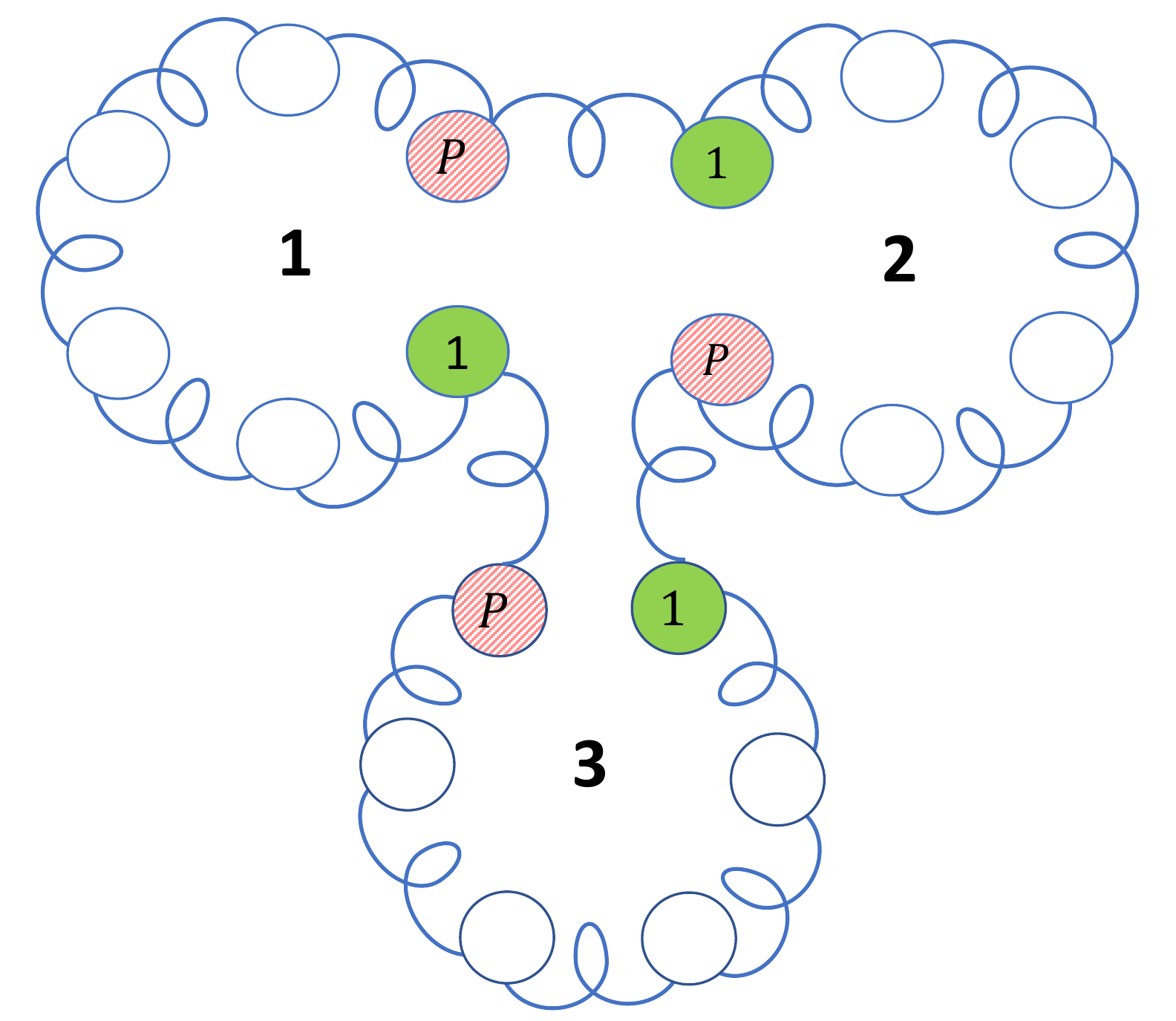}}
\end{minipage}%
\begin{minipage}{.33\textwidth}
\center
    \subfloat[$\permutationthree{3}{1}{2} = \cyclenotate{132}$]{\includegraphics[trim={3.5in 1in 3.5in 1in}, clip, width=0.8\textwidth]{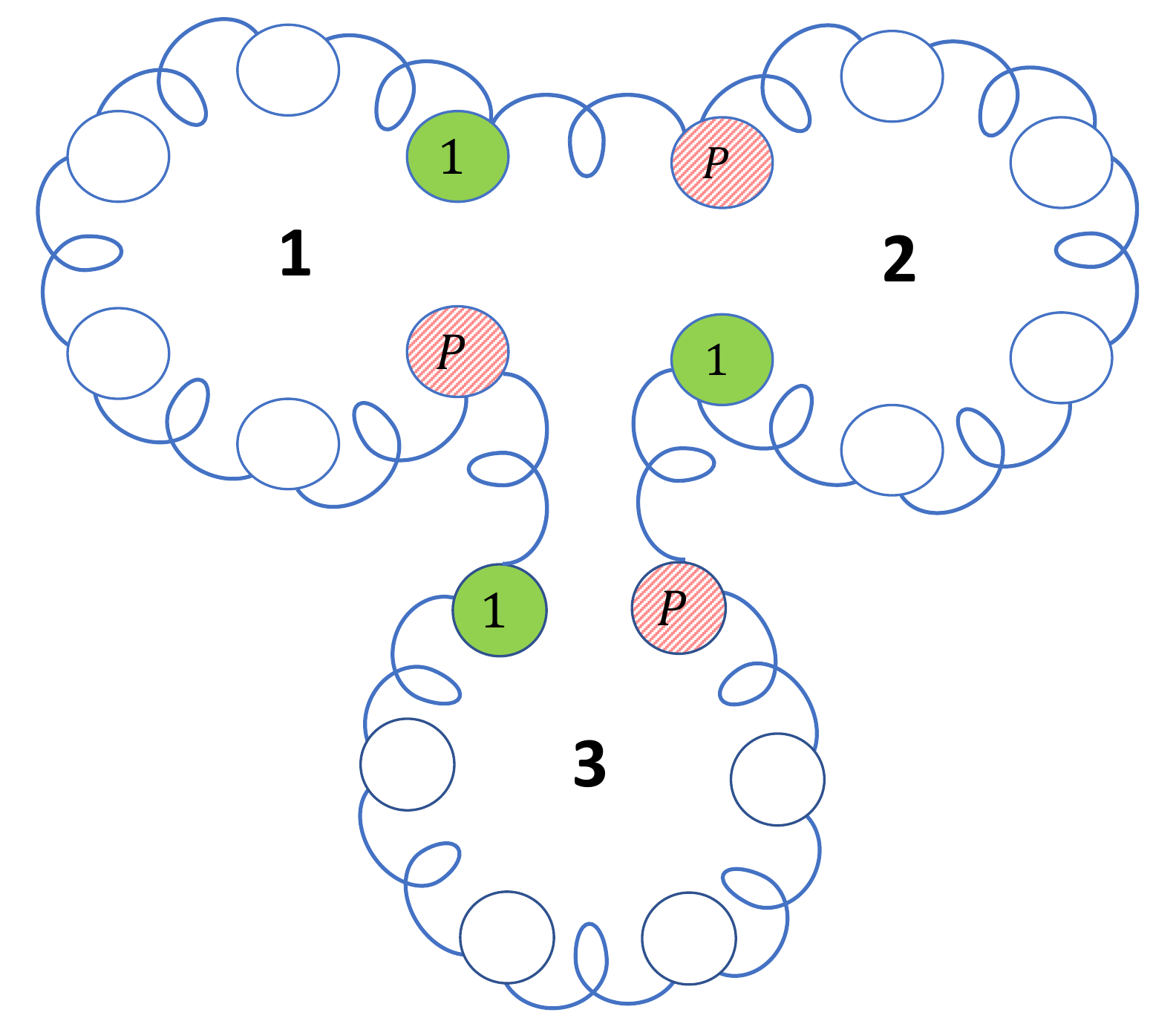}}
\end{minipage}
\\
\vspace{0.1in}
\begin{minipage}{.33\textwidth}
\center
    \subfloat[$\permutationthree{2}{1}{3} = \cyclenotate{12}\cyclenotate{3}$]{\includegraphics[trim={3.5in 1in 3.5in .9in}, clip, width=0.8\textwidth, cfbox=blue, frame]{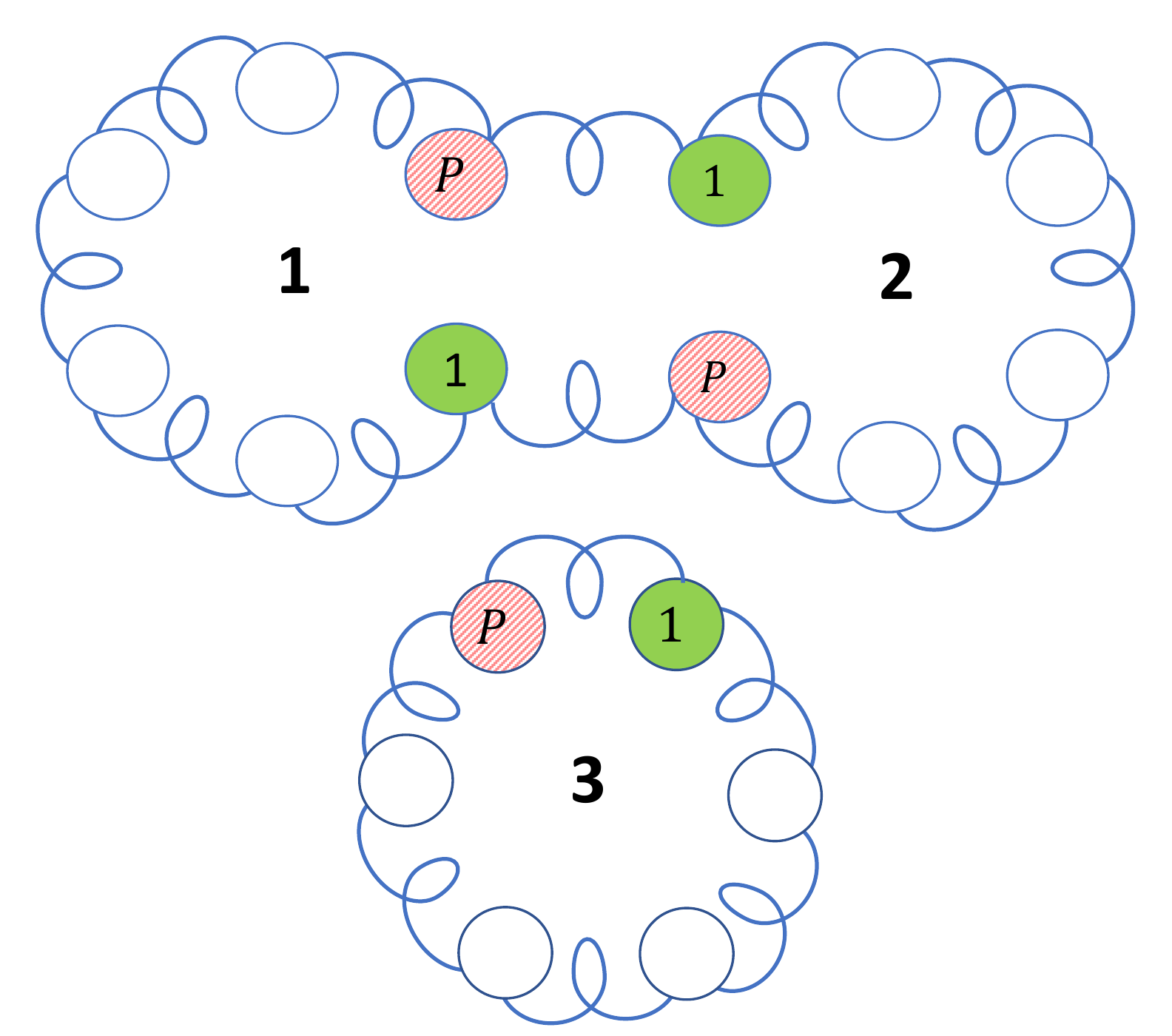}}
\end{minipage}%
\begin{minipage}{.33\textwidth}
\center
    \subfloat[$\permutationthree{3}{2}{1} = \cyclenotate{13}\cyclenotate{2}$]{\includegraphics[trim={3.5in 1in 3.5in .9in}, clip, width=0.8\textwidth]{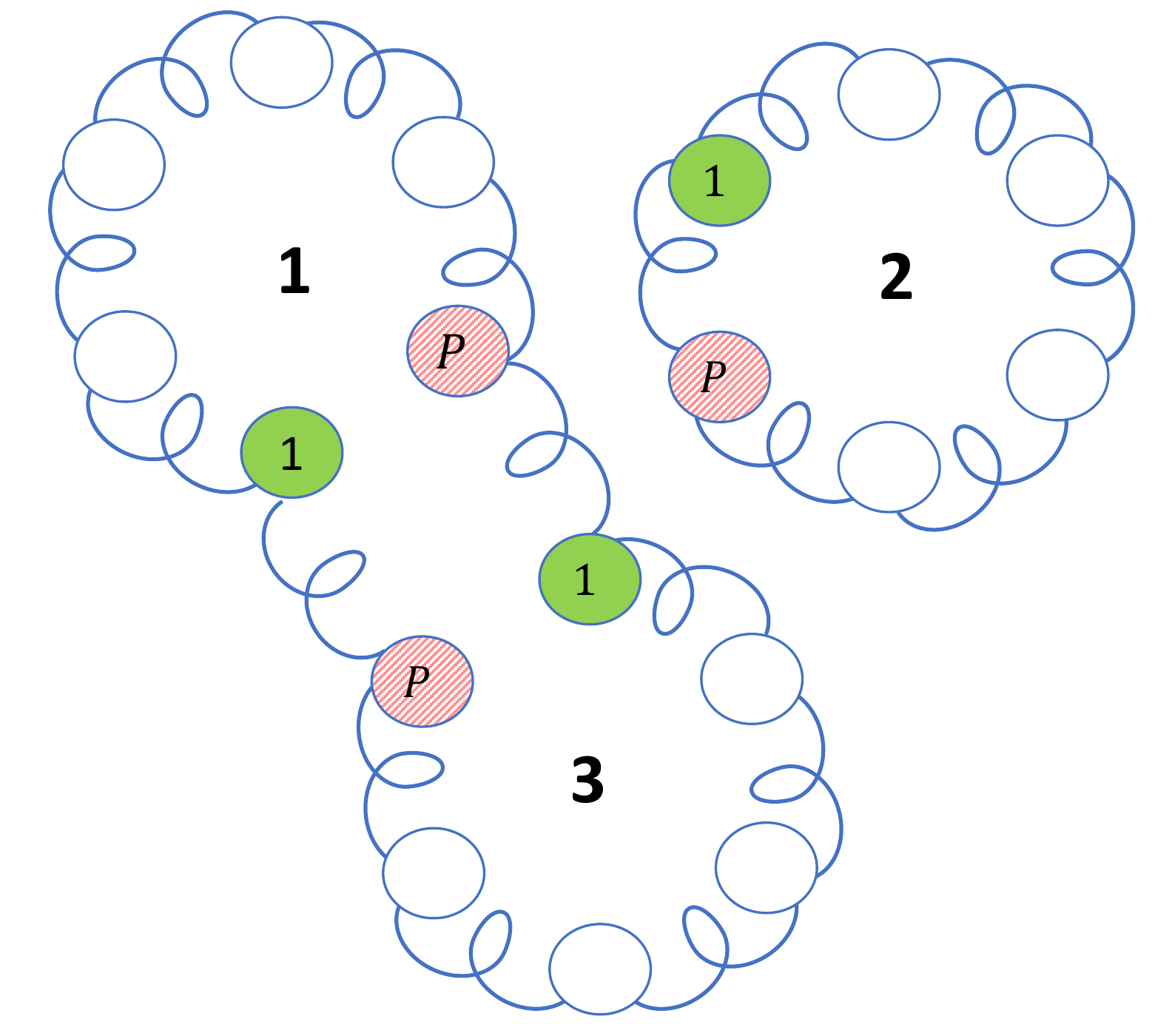}}
\end{minipage}%
\begin{minipage}{.33\textwidth}
\center
    \subfloat[$\permutationthree{1}{3}{2} = \cyclenotate{1}\cyclenotate{23}$]{\includegraphics[trim={3.5in 1in 3.5in .9in}, clip, width=0.8\textwidth, cfbox=blue, frame]{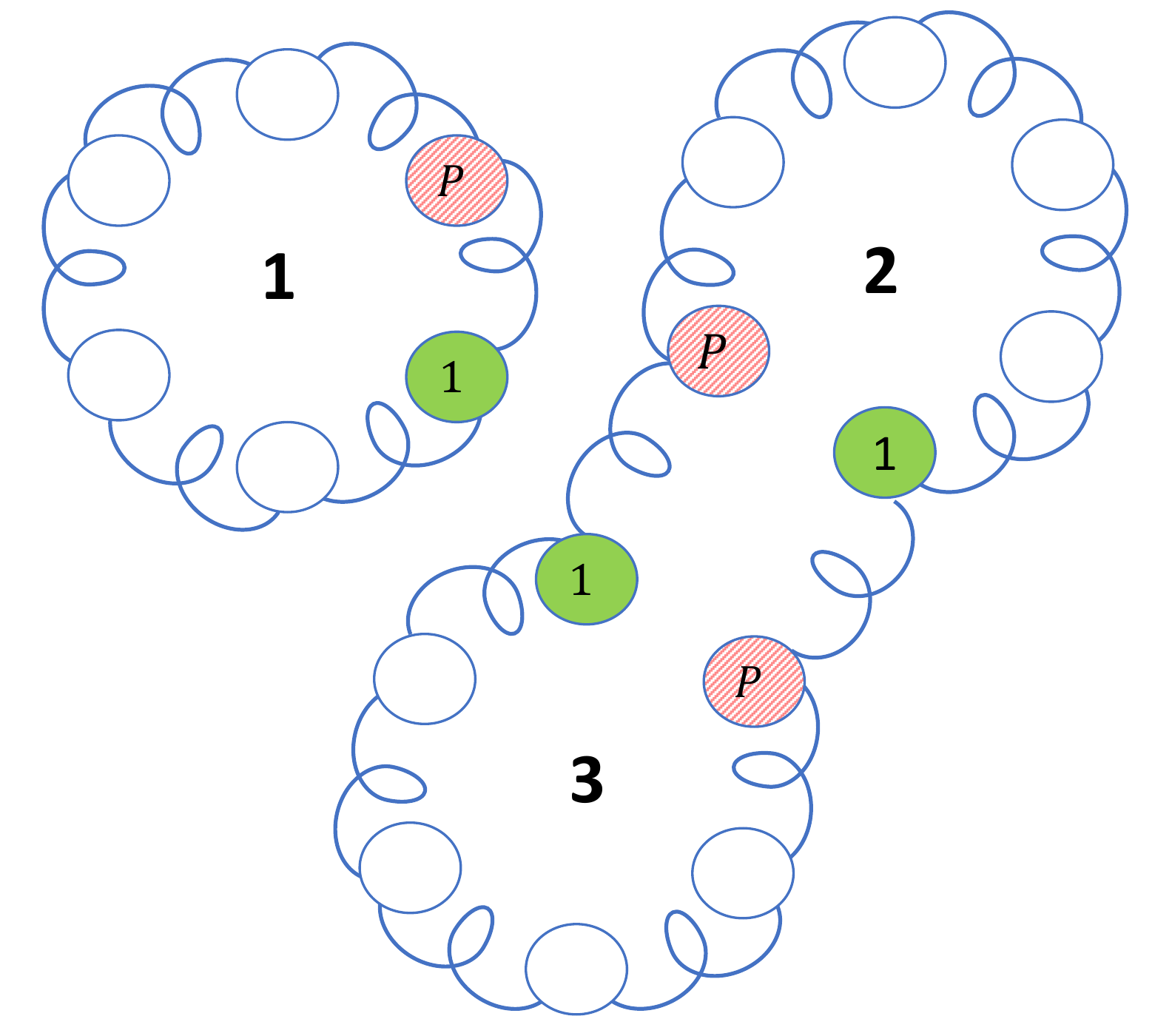}}
\end{minipage}
\end{minipage}
\caption{Ring polymer configurations for all permutations over three particles. 
The first and last beads are marked in green (solid) and red (dashed), respectively. 
Each permutation is given in two notations: the two-line notation, 
${\left(\begin{smallmatrix} 
1 & 2 & \ldots & N\\
{\sigma(1)} & {\sigma(2)} & \ldots & {\sigma(N)}
\end{smallmatrix}\right)}$,
and the cycle notation, where each cycle is a sequence $(\ell, \sigma(\ell), \sigma(\sigma(\ell)), \ldots)$, which must close due to the properties of permutations, and each element appears in exactly one cycle (see the supporting information, SI). %
Permutations that contribute to the potential $\Vall$ are enclosed in a square.}
\label{fig:all-permutations}
\end{figure*} 

\Cref{eq:int-of-sum-all-permutations} provides an expression for the bosonic partition function that can be sampled, at least in principle, through %
simulations of a classical system with ring polymer potential 
$-\frac{1}{\beta} \ln{
    \left[
        \frac{1}{\fact{N}} 
        \sum_{\sigma}{\boltzmann{\Eperm{\sigma}}} %
    \right]
    }
.$
However, the number of permutations $\fact{N}$ grows exponentially, making this naive approach intractable.

Recently, we showed~\cite{hirshberg2019path} that the exponential scaling can be avoided %
by sampling the ring polymer potential $\Vall = \Vall(\pos_1,\ldots,\pos_N)$ defined through the following recurrence relation:
\begin{equation}
\label{eq:our-forward-potential-recurrence}
\begin{alignedat}{2}
    &\boltzmann{\Vall} &&= \frac{1}{N} \sum_{k=1}^{N}{\boltzmann{\left(\Vto{N-k} + \Enk{N}{k}\right)}}.
\end{alignedat}
\end{equation}
The recursion is terminated by setting $\Vto{0} = 0$.
In \Cref{eq:our-forward-potential-recurrence},
$\Enk{N}{k} = \Enk{N}{k}(\pos_{N-k+1},...,\pos_N)$ is the spring energy of the ring polymer obtained by connecting the beads of particles $N-k+1,N-k+2,...,N$ in a cycle,
\begin{equation}
\Enk{N}{k} = \springforceprefix \sum_{\ell=N-k+1}^{N}{\sum_{j=1}^{P}{\rdiffsquared{\ell}{j+1}{\ell}{j}}},
\label{eq:Ek}
\end{equation}
where $\beadpos{\ell}{P+1} = \beadpos{\ell+1}{1}$ except for $\beadpos{N}{P+1} = \beadpos{N-k+1}{1}$. 
In the following, we refer to $\Enk{N}{k}$ as \emph{cycle energies}~\footnote{
    $\Vto{N}$ and $\Enk{N}{k}$ in our current notation correspond to $\originalpotentialnotation{N}$ and $E_N^{(k)}$ of Hirshberg~\textit{et al.}~\cite{hirshberg2019path}, respectively.
}.
The potential $\Vall$ is different from the potential obtained by the sum over all permutations, but it samples the bosonic partition function in an equivalent manner~\cite{hirshberg2019path}. 

The recurrence relation defining $\Vall$ allows to compute both the potential and the derived forces in $\bigO(PN^3)$ time.
To evaluate the \emph{potential}, %
first the $\bigO(N^2)$ cycle energies are evaluated, each of which costs $\bigO(PN)$ due to the double sum in~\Cref{eq:Ek}, resulting in $\bigO(PN^3)$ overall.
The \emph{forces} due to the potential, i.e., $-\beadderive{\ell}{j}{\Vall}$ for all $\ell$ and $j$, are also computed in $\bigO(P N^3)$ time using a recurrence relation that is obtained from~\Cref{eq:our-forward-potential-recurrence} by taking the derivative.
See the Supporting Information of ref.~\citenum{hirshberg2019path} for the full complexity analysis.

In this paper, we improve on this complexity and show that the potential and forces can be computed in $\bigO(N^2 + PN)$ time, leading to significantly faster bosonic PIMD simulations. As in the previous method, the algorithm first computes the potential and then the forces on all beads. %
We achieve a faster evaluation of the potential by a recursive formula also for the cycle energies $\Enk{N}{k}$,
and a faster evaluation of the forces by writing a probabilistic expression for them.

\subsubsection{Reducing the potential to $\bigO(N^2 + PN)$}
\label{sec:faster-potential}

\begin{figure*}[t]
    \center
    \sbox0{
    \captionsetup[subfloat]{labelformat=empty}
    \begin{minipage}{.38\textwidth}
    \center
      \begin{minipage}{\linewidth}
      \center
        \begin{minipage}{0.5\textwidth}
        \center
          \subfloat[\footnotesize $\Efromto{1}{1}$]{\includegraphics[trim={3.5in 3.4in 6.75in .9in}, clip, width=0.5\textwidth]{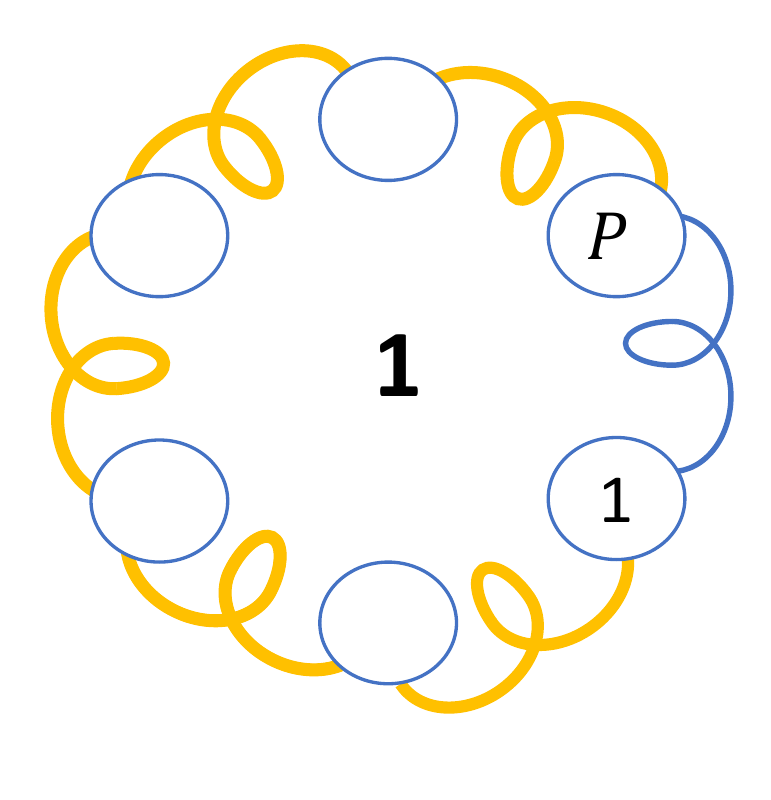}%
          }
        \end{minipage}%
        \begin{minipage}{0.5\textwidth}
        \center
          \subfloat[\footnotesize $\Efromto{1}{2}$]{\includegraphics[trim={3.5in 3.4in 3.5in .9in}, clip, width=\textwidth]{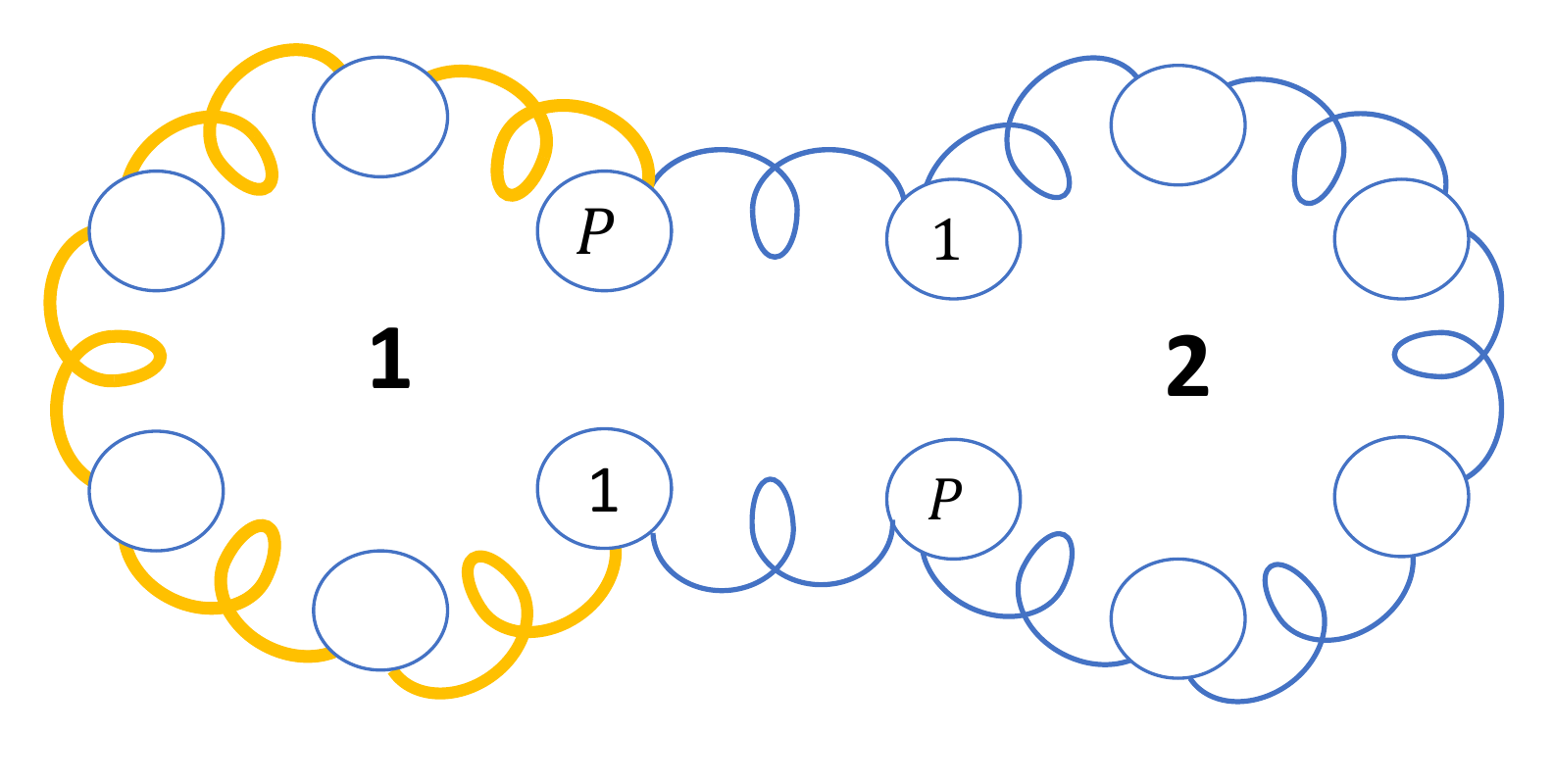}}
        \end{minipage}
      \end{minipage}
      \begin{minipage}{\linewidth}
      \center
          \subfloat[\footnotesize $\Efromto{1}{3}$]{\includegraphics[trim={3.5in 1in 3.5in .9in}, clip, width=0.5\textwidth]{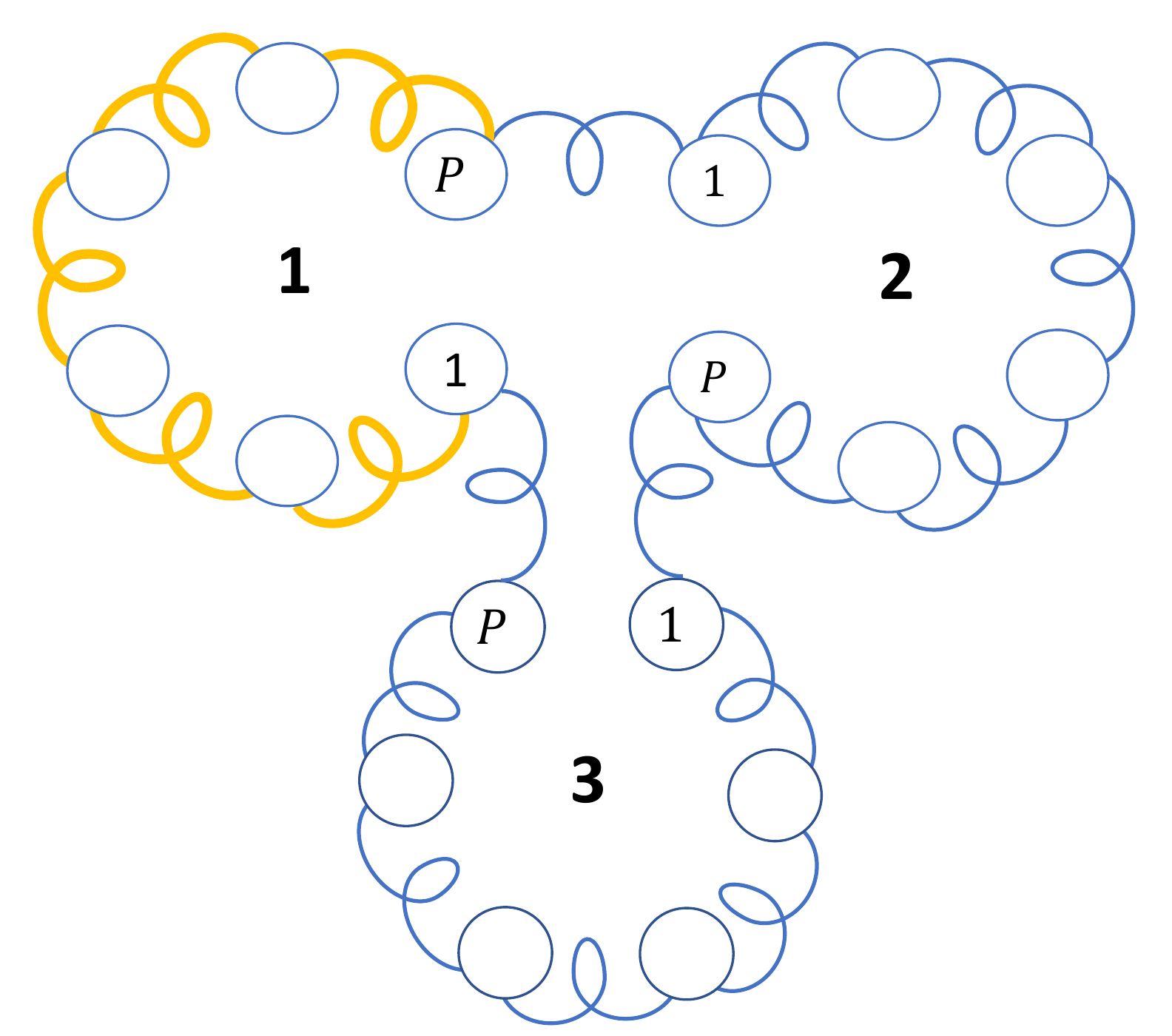}}
        \end{minipage}
    \end{minipage}%
    \setcounter{subfigure}{0}
    }
    \sbox1{
    \begin{minipage}{.38\textwidth}
    \center
      \includegraphics[trim={3.5in 1in 3.5in .9in}, clip, width=0.5\textwidth]{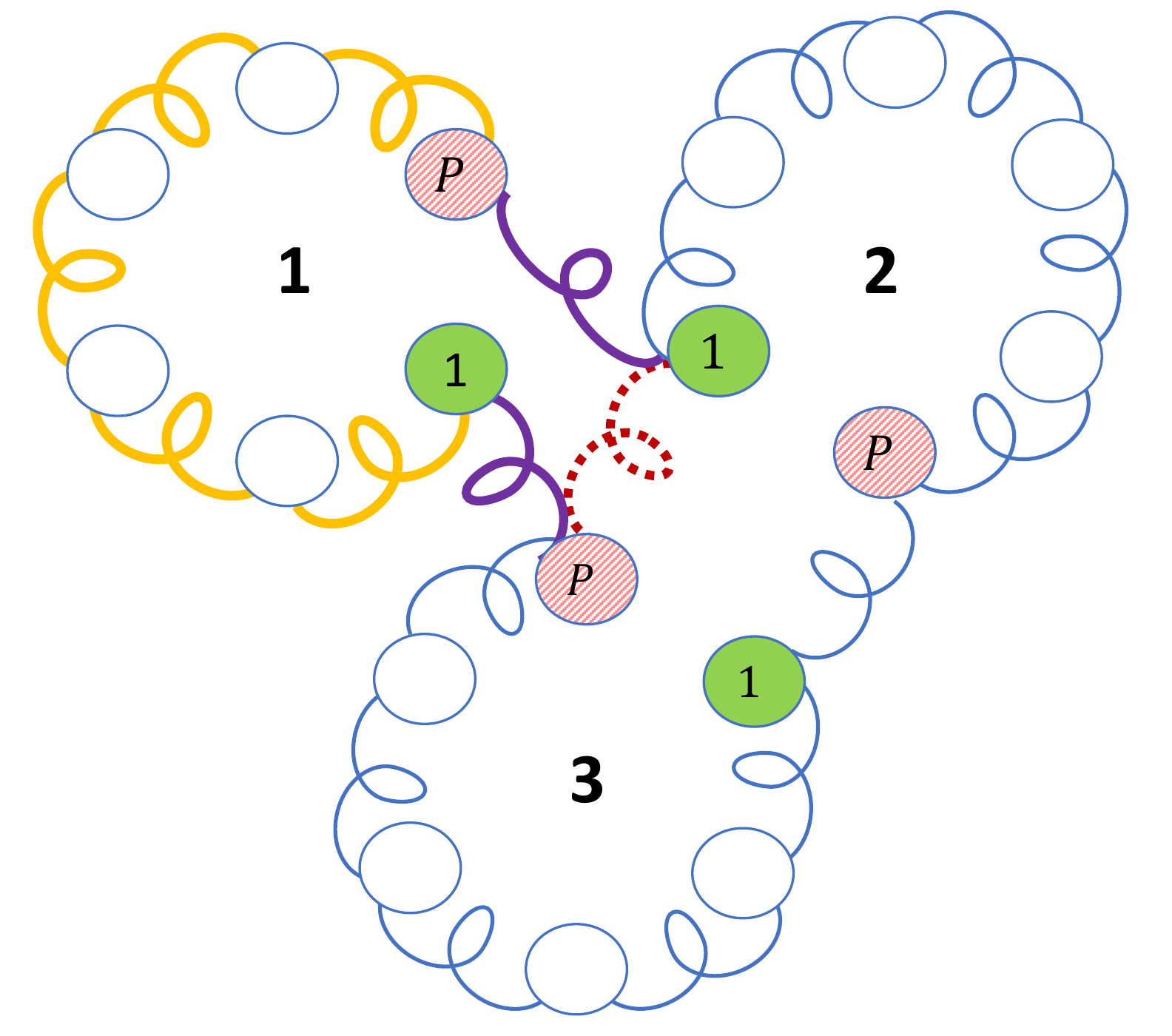}
    \end{minipage}
    }

\subfloat[\!\!\!\!\label{fig:cycle-base}]{\usebox0}
\subfloat[\label{fig:cycle-step}]{\usebox1}
\caption{
  (a) The interior springs of particle $1$ are the same in all the cycles in which the particle participates.
  (b) The cycle energy $\Efromto{1}{3}$ can be obtained from $\Efromto{2}{3}$ by removing the term corresponding to the the spring closing the cycle from $3$ to $2$ (red, dotted), connecting particle $1$ to the beginning of $2$ (purple, solid), adding the interior springs of particle $1$ (gold, solid), and closing the cycle from $3$ to $1$ (purple, solid).
}
\end{figure*}

 To reduce the computational cost of evaluating the potential, %
we evaluate the cycle energies by incrementally extending cycles one particle at a time.
We obtain each $\Efromto{\particledown}{\particleup}$ from $\Efromto{\particledown+1}{\particleup}$ by ``breaking'' and ``reforming'' the cycle:
\begin{equation}
    \label{eq:e-recurrence-step}
        \begin{aligned}
        \Efromto{\particledown}{\particleup} = \Efromto{\particledown+1}{\particleup}
                            &- \springforceprefix \rdiffsquared{\particleup}{P}{\particledown+1}{1}
                            \\
                            & + \springforceprefix \rdiffsquared{\particledown}{P}{\particledown+1}{1}
                            + \Einterior{\particledown}
                            \\
                            & + \springforceprefix \rdiffsquared{\particleup}{P}{\particledown}{1}.
        \end{aligned}
\end{equation}
In~\Cref{eq:e-recurrence-step}, we first open the ring by removing the spring between the last bead of $\particleup$ and the first bead of $\particledown+1$. We then connect $\particledown$ to $\particledown+1$, add the contribution from the interior springs of $\particledown$ ($\Einterior{\particledown}$ defined below), and then close the ring by connecting $\particleup$ to $\particledown$. 
\Cref{fig:cycle-step} illustrates this graphically.
The term $\Einterior{\particledown} = \Einterior{\particledown}(\pos_{\particledown})$ is defined by %
\begin{equation}
\label{eq:e-interior}
    \Einterior{\particledown} = \interparticleforce{\particledown},
\end{equation}
summing over the springs that are unaffected by particle exchange (see~\Cref{fig:cycle-base}).
Note that $j=P$ is not included in the summation.
Extending cycles begins with cycles that contain just one particle:
\begin{equation}
    \label{eq:e-recurrence-base}
    \Efromto{\particleup}{\particleup} = \Einterior{\particleup} + \springforceprefix \rdiffsquared{\particleup}{P}{\particleup}{1}.
\end{equation}

With this procedure, all the cycle energies are computed in $\bigO(N^2 + PN)$ time: 1) summing over the $P-1$ interior springs of $N$ particles (\Cref{eq:e-interior}) takes $\bigO(PN)$ time. 2) Forming the single-particle cycles of~\Cref{eq:e-recurrence-base} takes additional $\bigO(N)$, and 3) evaluating each cycle energy $\Efromto{\particledown}{\particleup}$ from $\Efromto{\particledown+1}{\particleup}$ (\Cref{eq:e-recurrence-step}) is $\bigO(1)$, with $\bigO(N^2)$ such cycles.

Once the cycle energies are known, evaluating $\Vall$ takes $\bigO(N^2)$ only, since it amounts to evaluating all the $N$ potentials $\Vto{\particleup}$, each of which costs $\bigO(N)$ due to the sum in~\Cref{eq:our-forward-potential-recurrence}.
\iflong
This part of the algorithm is shown in~\crefrange{ln:potential-start}{ln:potential-end} of~\Cref{alg:full-pseudocode} in~\Cref{sec:alg-pseduocode}.
\else
We provide a pseudocode of the potential evaluation algorithm in Algorithm 1 of the supporting information (SI). %
\fi

\subsubsection{Reducing the force to $\bigO(N^2 + PN)$}
Reducing the time required to compute the potential does not directly reduce the scaling of the evaluation of the forces. 
In the original algorithm, the recurrence
for the force is $\bigO(N^2)$ %
once the cycle energies are known. However, it needs to be evaluated separately per bead, which amounts to $\bigO(PN^3)$ overall.
The techniques we use above cannot improve this further.
To overcome this, we need to rewrite the forces in a new, probabilistic form.

\paragraph{Potential as a sum over representative permutations.}
To this end, we rewrite 
the potential of~\Cref{eq:our-forward-potential-recurrence}.
We show that %
it
can be equivalently defined by a sum over the permutations, provided that they are transformed by a function $\repsym$: %
\begin{equation}
\label{eq:vall-sum-representatives}
    \boltzmann{\Vall} = \frac{1}{\fact{N}} \sum_{\sigma}{\boltzmann{\Eperm{\rep{\sigma}}}}.
\end{equation}
Given a permutation $\sigma$, the function $\repsym$ returns a permutation $\rep{\sigma}$ that is called the \emph{representative permutation} of $\sigma$. %
For example, for three particles, we choose $\repsym$ such that $\rep{\cyclenotate{132}}=\cyclenotate{123}$, $\rep{\cyclenotate{13}\cyclenotate{2}}=\cyclenotate{1}\cyclenotate{23}$, and $\rep{\sigma}=\sigma$ otherwise (the permutations are identified through their cycle notation; see~\Cref{fig:all-permutations} and~\refappendix{sec:cycle-notation}). With this choice, \Cref{eq:vall-sum-representatives} for $N=3$ gives 
\begin{equation}
\label{eq:vall-3-sum-permutations}
\begin{alignedat}{2}
    \boltzmann{\Vto{3}} &= \frac{1}{6}\boltzmann{\Eperm{\cyclenotate{1}\cyclenotate{2}\cyclenotate{3}}} 
              &&+ \frac{1}{3}\boltzmann{\Eperm{\cyclenotate{123}}}
              \\
              &+ \frac{1}{6}\boltzmann{\Eperm{\cyclenotate{12}\cyclenotate{3}}}
              &&+ \frac{1}{3}\boltzmann{\Eperm{\cyclenotate{1}\cyclenotate{23}}},
\end{alignedat}
\end{equation}
which is exactly the result of expanding the recurrence of~\Cref{eq:our-forward-potential-recurrence}.
For a general $N$, $\rep{\sigma}$ is defined by first ordering the cycles of $\sigma$ is ascending order due to their largest element, then replacing all elements by $1,...,N$ consecutively.
The precise definition appears in~\hyperref[sec:methods]{Methods}.
The proof of~\Cref{eq:vall-sum-representatives} appears in
\iflong
\Cref{thm:forward-potential-recurrence}
\else
the SI.
\fi

The reason that~\Cref{eq:vall-sum-representatives}  transforms permutations using $\repsym$ is that $\Vall$ does \emph{not} include \emph{all} permutations. For example, \Cref{eq:vall-3-sum-permutations}
does not include terms for the permutations $\cyclenotate{132}$ and $\cyclenotate{13}\cyclenotate{2}$  
(see%
~\Cref{fig:all-permutations}).
However, \Cref{eq:vall-sum-representatives} proves that $\Vall$ indeed samples the correct quantum partition function, since $\rep{\sigma}$ preserves the topology of $\sigma$ by construction. 

It might seem counter-intuitive to rewrite the recursive potential as a sum over representative permutations, because their number also grows exponentially (as shown in~\refappendix{sec:combinatorics}).
However, we never use this expression to evaluate the potential, but rather %
to derive an efficient expression for the force, as we now explain.

\paragraph{Force as an expectation value.}
Since the potential is an average over representative permutations, per~\Cref{eq:vall-sum-representatives}, taking the derivative yields that
the force is a \emph{weighted average} over the representative permutations:
\begin{equation}
\label{eq:force-expected-force}
    \beadforce{\ell}{j}{\Vall} = \sum_{\sigma}{\Prerepperm{\sigma} \cdot \beadforce{\ell}{j}{\Eperm{\rep{\sigma}}}},
\end{equation}
where
\begin{equation}
\label{eq:representative-distribution}
    \Prerepperm{\sigma} = \frac{\boltzmann{\Eperm{\rep{\sigma}}}}{\sum_{\sigma}{\boltzmann{\Eperm{\rep{\sigma}}}}} = 
    \frac{\boltzmann{\Eperm{\rep{\sigma}}}}{\fact{N} \cdot \boltzmann{\Vall}}.
\end{equation}
For given positions %
of the beads, \Cref{eq:representative-distribution} defines a Boltzmann distribution of permutations. %
\Cref{eq:force-expected-force} states that the force on each bead $\beadpos{\ell}{j}$ is the \emph{expectation value} of the force over this distribution.

In itself, this expression is still not amenable to efficient computation due to the sum over the permutations. 
However, it can be greatly simplified by observing that, in any permutation,
the force on a bead depends only on its immediate neighbors.
When taking the expectation value, %
we can group together all permutations that have the same force exerted on the bead. %
The calculation then splits to two cases: the case of \emph{interior beads} ($j=2,\ldots,P-1$), and the case of \emph{exterior beads} ($j=1,P$).

\paragraph{Force on interior beads.}
The force on each interior bead is the same in all %
permutations: $\beadpos{\ell}{j}$ is always connected to $\beadpos{\ell}{j+1},\beadpos{\ell}{j-1}$, for $j = 2, \ldots, P-1$ and any $\ell$ (see~\Cref{fig:cycle-base}).
Therefore, its expectation value is trivial, given by the standard PIMD expression,
\begin{equation}
\label{eq:intermediate-bead-force-main-text}
    \beadforce{\ell}{j}{\Vall} = -\mass \springconstant^2 (2\posbead{\ell}{j} - \posbead{\ell}{j+1} - \posbead{\ell}{j-1}).
\end{equation}
This expression requires $\bigO(1)$ time to evaluate, resulting in $\bigO(PN)$ time %
for all the interior beads. 
Thus, our algorithm computes the force on most beads the same way as in ordinary PIMD, which already reduces much of the added cost of exchange effects. But to achieve quadratic scaling of the entire algorithm, we must compute the force on exterior beads efficiently as well.

\paragraph{Force on exterior beads.}
For beads $j=1,P$ of each particle $\ell$, different permutations \emph{do} exert different forces, depending on its neighbors. We %
group together all permutations in which the neighbor is the same:
\begin{align}
\label{eq:force-on-1-sum-neighbors}
    &\beadforce{\ell}{1}{\Vall} =
    \\
    \nonumber &\qquad -\sum_{\ell'=1}^{N}\Prrepnext{\ell'}{\ell}\cdot \mass \springconstant^2 (2\posbead{\ell}{1} - \posbead{\ell}{2} - \posbead{\ell'}{P}),
    \\
    \intertext{and}
\label{eq:force-on-P-sum-neighbors}
    &\beadforce{\ell}{P}{\Vall} = 
    \\
    \nonumber &\qquad -\sum_{\ell'=1}^{N}{\Prrepnext{\ell}{\ell'} \cdot \mass \springconstant^2 (2\posbead{\ell}{P} - \posbead{\ell'}{1} - \posbead{\ell}{P-1})}.
\end{align}
In these expressions, $\Prrepnext{\ell}{\ell'}$ represents the probability that bead $P$ of particle $\ell$ is connected to bead $1$ of particle $\ell'$.
It is defined by summing the probabilities of~\Cref{eq:representative-distribution} over all permutations in which this condition holds in $\rep{\sigma}$.

To demonstrate~\Cref{eq:force-on-P-sum-neighbors} for $N=3$, the force on bead $P$ of particle $1$ is
\begin{align*}
    \beadforce{1}{P}{\Vto{3}} = 
        &-\Prrepnext{1}{1} \cdot \mass \springconstant^2 (2\posbead{1}{P} - \posbead{1}{1} - \posbead{1}{P-1})
        \\
        &-\Prrepnext{1}{2} \cdot \mass \springconstant^2 (2\posbead{1}{P} - \posbead{2}{1} - \posbead{1}{P-1})
        \\
        &-\Prrepnext{1}{3} \cdot \mass \springconstant^2 (2\posbead{1}{P} - \posbead{3}{1} - \posbead{1}{P-1})
\end{align*}
where
\begin{equation*}
\begin{alignedat}{2}
\Prrepnext{1}{1} \, &\propto \,
                          \frac{1}{6} \cdot \boltzmann{\Eperm{\cyclenotate{1}\cyclenotate{2}\cyclenotate{3}}}
                          &&+
                          \frac{1}{3} \cdot \boltzmann{\Eperm{\cyclenotate{1}\cyclenotate{23}}}
                          ,
\\
\Prrepnext{1}{2} \, &\propto \,  
                          \frac{1}{6} \cdot \boltzmann{\Eperm{\cyclenotate{12}\cyclenotate{3}}} 
                          &&+ 
                          \frac{1}{3} \cdot \boltzmann{\Eperm{\cyclenotate{123}}},
\\
\Prrepnext{1}{3} \, &= \, 0 &&
\end{alignedat}
\end{equation*}
(see~\Cref{eq:vall-3-sum-permutations} and~\Cref{fig:all-permutations}).

\Cref{eq:force-on-1-sum-neighbors,eq:force-on-P-sum-neighbors} imply that once the connection probabilities $\Prrepnext{\ell}{\ell'}$ are known for all $\ell,\ell'$, %
the force on all exterior beads can be computed in $\bigO(N^2)$ time: there are $2N$ of them, and each is computed in~$\bigO(N)$ time due to the sum in~\Cref{eq:force-on-1-sum-neighbors,eq:force-on-P-sum-neighbors}.
Our final goal then is to precompute the connection probabilities efficiently. %
This is a mighty challenge, because the connection probability sums over %
an exponential number of permutations
(as shown in~\refappendix{sec:combinatorics}).

\paragraph{Connection probabilities.}
To derive efficient expressions for the connection probabilities, 
we define another set of potentials, $\Vfrom{N},\Vfrom{N-1},\ldots,\Vfrom{1}$ through a similar recursion relation,
\begin{equation}
\label{eq:backward-potential-recursion-main-text}
    \boltzmann{\Vfrom{u}} = \sum_{\ell=u}^{N}{\frac{1}{\ell} \boltzmann{\left(\Efromto{u}{\ell} + \Vfrom{\ell+1}\right)}}.
\end{equation}
Note that the first index changes, as opposed to~\Cref{eq:our-forward-potential-recurrence}, and that the recursion is terminated by setting $\Vfrom{N+1} = 0$.
As before, these potentials can be computed in $\bigO(N^2)$, since we already have the cycle energies.

Once the potentials $\Vfrom{\particledown}$ are known, 
we show that all the connection probabilities can be computed in $\bigO(N^2)$ time: there are $N^2$ such probabilities,
and we provide expressions that compute each one in $\bigO(1)$ time. %
The intuition is that $\Vfrom{\ell+1}$ accounts for all the permutations over the particles $\ell+1,\ldots,N$, similarly to how $\Vto{\ell}$ accounts for permutations over $1,\ldots,\ell$. 
In this way, for example, we can express the connection probability for $\ell'=\ell+1$ as
\begin{equation}
\label{eq:direct-link-probability-main-text}
\begin{alignedat}{1}
    &\Prrepnext{\ell}{\ell+1} = 
    \\
    &\qquad \qquad 1 - \frac{1}{\boltzmann{\Vall}} \boltzmann{\left(\Vto{\ell} + \Vfrom{\ell + 1}\right)},
\end{alignedat}
\end{equation}
which is the complement of the probability that $\ell,\ell+1$ \emph{do not} belong to the same cycle in representative permutations.
Expressions for the rest of the connection probabilities appear in~\hyperref[sec:methods]{Methods}.
We provide a pseudocode of the force evaluation algorithm in Algorithm 1 of the SI.
\subsection{Benchmarks \& Applications}
\label{sec:empirical}

This section presents our numerical results for the improved performance of the algorithm.

\begin{figure*}
\center
\subfloat[\!\!\!\!\!\!\!\!\!\!\!\!\!\!\!\!\!\!\!\!\!\!\!\!\!\!\label{fig:speedup-scaling}]{\includegraphics[width=0.42\textwidth]{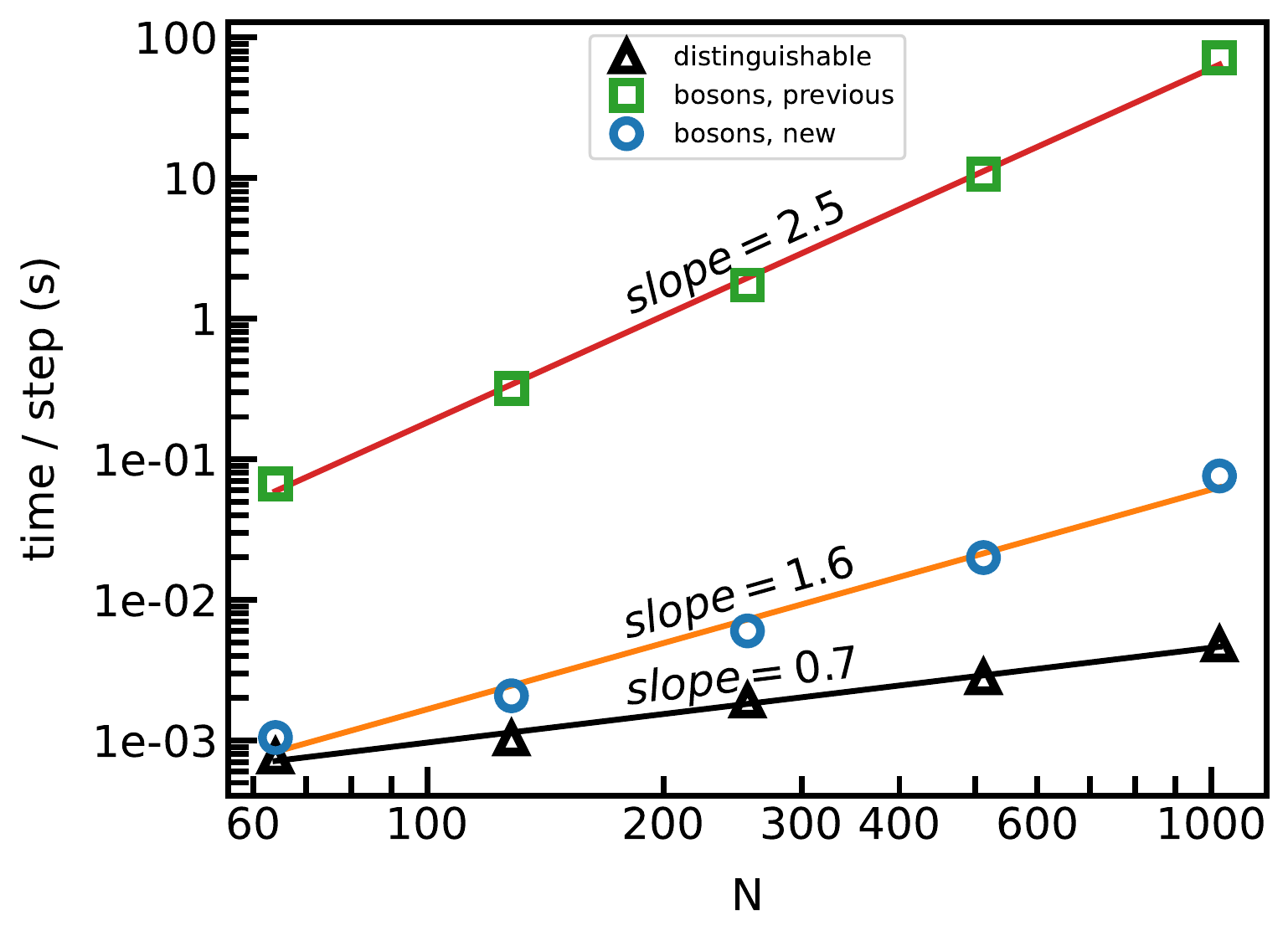}}%
\hspace{0.1in}
\subfloat[\!\!\!\!\!\!\!\!\!\!\!\!\!\!\!\!\!\!\!\!\!\!\!\!\!\!\label{fig:speedup-speedup}]{\includegraphics[width=0.42\textwidth]{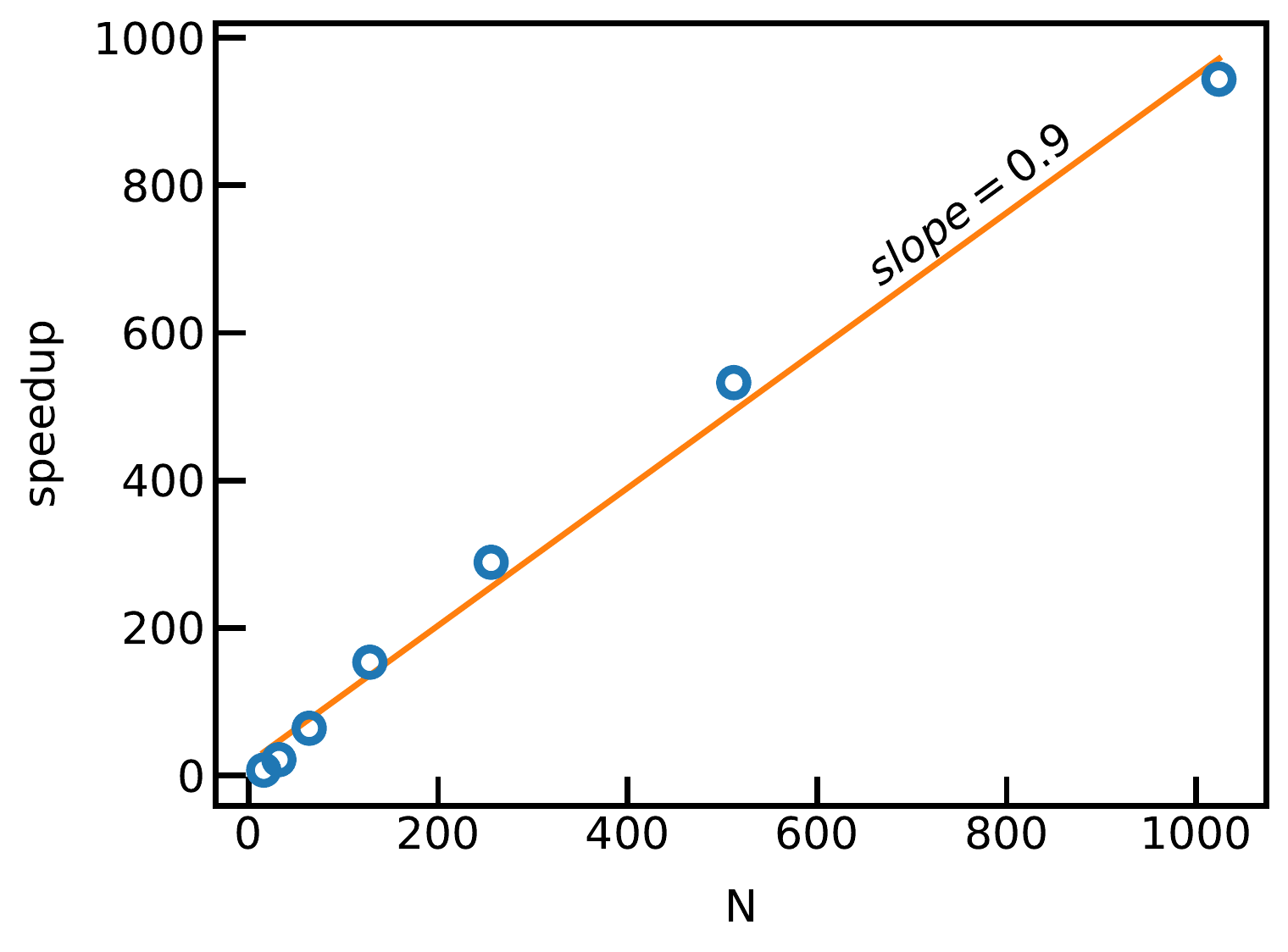}}%
\caption{
\label{fig:speedup}
  (a) The scaling of the original and current algorithms with the number of particles for a system of $N$ non-interacting cold trapped bosons on a log-log scale. Blue circles and green squares represent the current and original algorithms, respectively. Black triangles represent simulations of distinguishable particles. Lines represent a linear fit to the simulation data. (b) The speedup gained by using the proposed algorithm, defined as the ratio between time per step in the original and current implementations.  %
}
\end{figure*}

\begin{figure}
\center
    \includegraphics[width=0.42\textwidth]{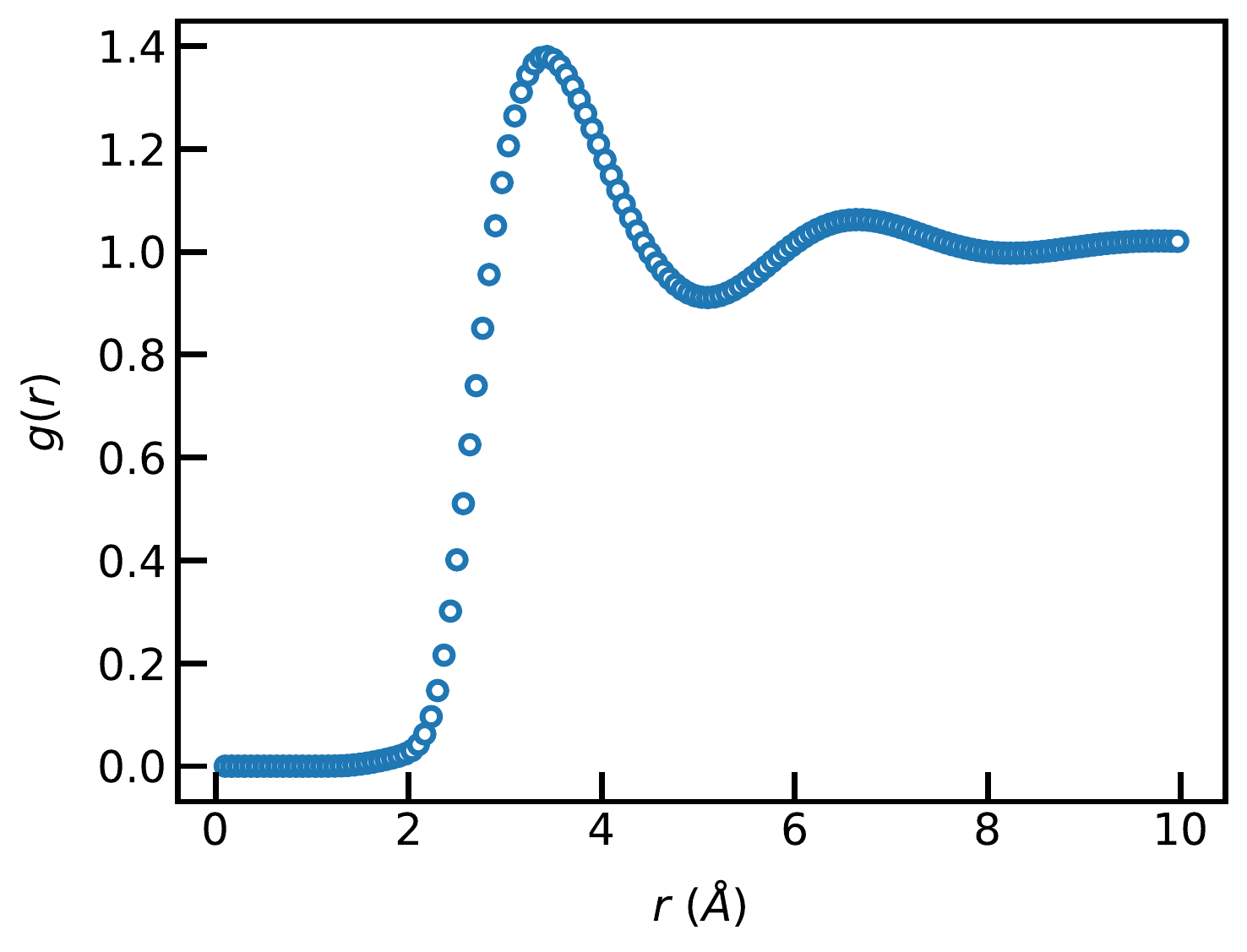}
\caption{
\label{fig:liquid-he}
  The pair correlation of $N=1372$ \ce{^{4}He} atoms in the superfluid liquid phase at temperature $1.379$K with $P=96$ beads. %
  The simulation lasts $\leq 7$ days, instead of $> 14$ years with the previous algorithm.
}
\end{figure} 
\paragraph{Scaling.}
The improved scaling of our algorithm, quadratic rather than cubic, is shown numerically in~\Cref{fig:speedup-scaling} for our LAMMPS~\cite{LAMMPS} implementation. Here, we measure the time per simulation step as a function of the number of particles $N$ on a log-log scale, for a system of trapped cold non-interacting bosons with $P=32$ beads.
The slope of the linear fit
is $1.6$ for our new algorithm (blue), compared to $2.5$ for the original algorithm (green). %
In our i-PI~\cite{KAPIL2019214} implementation, the slopes are $1.9$ and $3.2$ respectively 
\iflong
(see~\refappendix{sec:i-pi-scalability}). 
\else
(see the SI).
\fi

\paragraph{Speedup.}
The reduced scaling of the current algorithm implies a speedup of $\bigO(PN)$ over the original algorithm.
\Cref{fig:speedup-speedup} shows the speedup (the ratio of the time per step between the original and current algorithms) in the same system as in panel (a).
The new algorithm achieves acceleration by orders of magnitude, ranging from $\times 21$ faster simulations of $N=32$ bosons to $\times 944$ speedup for $N = 1024$.
The speedup is close to linear with the number of particles ($R^2=0.993$), with a slope of $\mathord{\sim} 0.9$.
We note that the cost of evaluating pairwise interactions, if present, is the same in the current and original algorithm (and in ordinary PIMD). This reduces the speedup in such systems, but it remains significant; we find a linear speedup with slope $0.8$ in liquid \ce{^{4}He}, $\times 6$ for $N=32$, and $\times 1069$ for $N=1372$ (see~\refappendix{sec:benchmarks-supplementary}).

\paragraph{Simulating $> 1000$ particles.}
Such speedups allow us to simulate systems far beyond what was previously possible in bosonic PIMD.
\Cref{fig:liquid-he} shows a simulation of superfluid liquid \ce{^{4}He}, where we compute the pair correlation of $N=1372$ atoms, with $P=96$ beads. %
The pair correlation is in very good agreement with the PIMC results of Ceperley~\cite{ceperley1995path} (see~\refappendix{sec:benchmarks-supplementary}).
This simulation takes $\mathord{\sim} 7$ days in our algorithm, but would take \emph{longer than 14 years} with the original algorithm (see~\hyperref[sec:methods]{Methods}).
Similarly, simulations of $N=1600$ trapped cold bosons, with and without Gaussian repulsive interaction (see~\refappendix{sec:benchmarks-supplementary}), take less than $\mathord{\sim} 9$ days, whereas the previous algorithm is estimated to require \emph{over 20 years} to complete on the same hardware (see~\hyperref[sec:methods]{Methods}).

\paragraph{Bosonic vs.\ ordinary PIMD.}
The new algorithm eliminates most of the added computational cost of bosonic exchange, and brings the scale of bosonic PIMD to be largely on par with ordinary PIMD.
\Cref{fig:speedup-scaling} compares the new algorithm to the PIMD algorithm in LAMMPS for distinguishable particles (black). 
For $64$ particles, there is only very little difference ($37\%$) in the time per simulation step. For $1024$ particle, bosonic simulations are merely $\mathord{\sim} 16$ times slower than for distinguishable particles. In the original algorithm, bosonic simulations of $1024$ particles would have been $\mathord{\sim} 15000$ times more expensive than simulations of distinguishable particles.
\section{Conclusions}
\label{sec:conclusions}
This paper presents an algorithm for bosonic PIMD with improved computational complexity, reducing the cost from cubic to quadratic with the number of particles. %
Our improved algorithm relies on two observations: 1) that the interior beads of each ring polymer are not affected by exchange interaction, and 2) that the bosonic ring polymer force can be evaluated efficiently as the expectation value of the force over a Boltzmann distribution of particle permutations. %

Using our method, simulations of condensed phase bosonic systems with thousands of particles can be done in days instead of years, allowing the first PIMD simulations of thousands of interacting particles including bosonic exchange symmetry.

We believe that our algorithm is a significant step towards including bosonic exchange effects in PIMD simulations as routinely as performing standard simulations for distinguishable particles.
The new algorithm will enable studies of systems such as exotic phases of superfluids under confinement using PIMD. 
Our insights might also be used to accelerate other algorithms where a similar recursion has been identified~\cite{PhysRevResearch.2.043206,10.1063/5.0026606,PhysRevE.107.055302}.
Finally, it would be exciting to see whether approximate methods for obtaining time correlation functions based on PIMD, such as Matsubara dynamics, Ring Polymer Molecular Dynamics or Centroid Molecular Dynamics~\cite{Althorpe2021}, could be devised also for bosonic systems. 
\section{Methods}
\label{sec:methods}
\subsection{Simulation details}
We implemented our algorithm in development branches of LAMMPS~\cite{LAMMPS} and i-PI~\cite{KAPIL2019214}. 
Unless otherwise stated, the results are obtained using the implementation in LAMMPS.
We validated the proposed algorithm by reproducing results obtained before using the original method~\cite{hirshberg2019path} on several bosonic systems with and without interaction, and extending some of them to larger system sizes, as detailed in~\refappendix{sec:benchmarks-supplementary}.

For the trapped bosons presented in~\Cref{fig:speedup}, we use $P=36$, a time step of $1$ fs, and average the time over $1000$ simulation steps. %
For the superfluid liquid \ce{^{4}He} presented in~\Cref{fig:liquid-he}, the simulation was performed with the interaction potential of Aziz~\textit{et al.}~\cite{Aziz79}, at a temperature of $1.379$K %
and a density of $0.02182 \AA^{-3}$, in periodic boundary conditions in all spatial dimensions, and for $2.5 \cdot 10^6$ steps of $0.5$ fs each, where the first 20\% were used as equilibration. 
For the two-dimensional density of trapped bosons with and without Gaussian repulsive interaction (whose results are displayed in~\refappendix{sec:benchmarks-supplementary}) we use $P=36$, appropriate for the temperature of $11.6$K ($\beta \hbar \omega=3$), and a time step of $1$ fs. The simulations were performed for $3 \cdot 10^6$ steps with the first 20\% discarded as equilibration. The repulsive potential is the same as described in ref.~\citenum{hirshberg2019path}, with $g=3$.

\subsection{Estimating the run-time of the original algorithm}
To estimate the time required for the implementation of the original algorithm to complete the simulation of $N=1600$ trapped cold bosons with a Gaussian repulsive interaction, we simulated $1000$ steps of the implementations of both the current and original algorithm, and extrapolated to $3 \cdot 10^6$ steps. The result was in agreement with the actual time required to perform the simulation in the new method, and actually slightly shorter, rendering our prediction conservative. Specifically, with the new method, $1000$ steps lasted $\mathord{\sim} 200$ seconds for $g=0$ and $\mathord{\sim} 150$ seconds for $g=3$, which is extrapolated to $\mathord{\sim} 7$ days for $g=0$ and $\mathord 5$ days for $g=3$; our actual simulation was slightly slower, $\mathord{\sim} 9$ days for $g=0$ and $\mathord 7$ days for $g=3$, which is likely due to the presence of other CPU-heavy processes on the same machine. In contrast, with the original method, $1000$ steps took $\mathord{\sim} 80$ hours for $g=0$ and $\mathord{\sim} 64$ hours for $g=3$, which is extrapolated to $\mathord{\sim} 27$ years and $\mathord{\sim} 22$ years, respectively, for $3 \cdot 10^6$ steps.

We estimated the time to simulate the system of superfluid \ce{^{4}He} in a similar manner.
With the current algorithm, simulating $1000$ steps on $N=1372$ and $P=64$ took $\mathord{\sim} 171$ seconds, which is extrapolated to $\mathord{\sim} 5$ days. %
In contrast, with the original algorithm, $1000$ steps required $\mathord{\sim} 51$ hours with $P=64$, which is extrapolated to $\mathord{\sim} 14.5$ years.
This is a conservative estimate for the simulation with the original algorithm and $P=96$, because increasing the number of beads increases the required time.
We did not perform this estimation based on the time to perform $1000$ steps with $P=96$ because the original implementation uses a simple numerical stability procedure when evaluating sums of exponentials (see~\refappendix{sec:numerical-stability}), %
which is not sufficiently stable for large $P$ and $N$.

\subsection{Definition of $\rep{\sigma}$}
\label{sec:def-repsym}
In the context of rewriting the potential as a sum over representative permutations (\Cref{eq:vall-sum-representatives}),
for a general $N$, the representative permutation $\rep{\sigma}$ is defined by transforming the cycle notation of $\sigma$, in two steps. First, the cycles are sorted in ascending order according to the largest element in each cycle.
For example, $\sigma = \cyclenotate{13}\cyclenotate{2}$ is rewritten as $\cyclenotate{2}\cyclenotate{13}$, because the highest element in the cycle $\cyclenotate{13}$ is larger than the highest element in $\cyclenotate{2}$. Second, the elements are replaced by the numbers $1,\ldots,N$ consecutively while maintaining the same length and order of cycles. In the same example, $\cyclenotate{2}\cyclenotate{13}$ turns into $\cyclenotate{1}\cyclenotate{23}$, and overall $\rep{\cyclenotate{13}\cyclenotate{2}} = \cyclenotate{1}\cyclenotate{23}$.

\subsection{Expressions for connection probabilities}
\label{sec:connection-probabilities-all-exprs}
In the context of evaluating the forces (\Cref{eq:force-on-1-sum-neighbors,eq:force-on-P-sum-neighbors}), we describe how to compute the connection probabilities based on the potentials defined in~\Cref{eq:backward-potential-recursion-main-text}.

Consider the connection probability $\Prrepnext{\ell}{\ell'}$.
For $\ell' > \ell$, the probability is nonzero only when particle $\ell$ is not the particle with the highest index in the ring, due to the properties of $\rep{\sigma}$ (see~\Cref{fig:all-permutations}). In that case, $\ell$ is connected to the particle $\ell+1$.
We prove 
\iflong
in~\Cref{lem:direct-link-probability} (\Cref{sec:proofs})
\else
in the SI 
\fi
that the expression for this probability is
the expression in~\Cref{eq:direct-link-probability-main-text} above.

On the other hand, a connection $\ell' \leq \ell$ happens only when particle $\ell$ is the particle with highest index in the ring.
We prove in~\refappendix{lem:close-cycle-probability}
\iflong
(\Cref{sec:proofs})
\else
\fi
that for $\ell' \leq \ell$,
\begin{equation}
\label{eq:close-cycle-probability-main-text}
\begin{alignedat}{1}
    &\Prrepnext{\ell}{\ell'} = 
    \\
    &\qquad \qquad \frac{1}{\ell} \frac{1}{\boltzmann{\Vall}} {\boltzmann{\left(\Vto{\ell'-1} + \Efromto{\ell'}{\ell} + \Vfrom{\ell+1}\right)}}. 
\end{alignedat}
\end{equation}

The intuition underlying the proof is that $\Vfrom{\ell+1}$ accounts for all the permutations over the particles $\ell+1,\ldots,N$, similarly to how $\Vto{\ell}$ accounts for the permutations over $1,\ldots,\ell$.
Therefore, in~\Cref{eq:direct-link-probability-main-text}, we compute the complement of the probability that $\ell,\ell+1$ \emph{do not} belong to the same cycle.
Similarly, in~\Cref{eq:close-cycle-probability-main-text}, we account for all configurations that include a ring polymer that starts with particle $\ell'$ and ends with particle $\ell$, and all possible permutations over the other particles. 

\begin{acknowledgments}
Barak Hirshberg acknowledges support by the USA-Israel Binational Science Foundation (grant No. 2020083) and the Israel Science Foundation (grants No.\ 1037/22 and 1312/22). Yotam Feldman was supported by the Ratner Center Fellowship, as well as Schmidt Science Fellows, in partnership with the Rhodes Trust.
\end{acknowledgments}

\bibliography{refs}%

\iflong
\onecolumngrid
\clearpage
\appendix

\newtheorem{theorem}{Theorem}
\newtheorem{lemma}{Lemma}
\newtheorem{definition}{Definition}

\section{Algorithm Pseudocode}
\label{sec:alg-pseduocode}
\begin{algorithm}
\DontPrintSemicolon
$\label{ln:potential-start} \label{ln:interior-energy-start}$\tcp{compute cycle energies $\Enk{\particledown}{\particleup}$}
\For{$\particleup = 1 \ldots N$} {
	$\begin{aligned}
	\label{ln:enk-internal-energy}
	\Einterior{\particleup} = \interparticleforce{\particleup}.
	\label{ln:interior-energy-end}
	\end{aligned}$ \;
}
$\label{ln:cycle-energies-start}$\For{$\particleup = 1 \ldots N$} {
$\Efromto{\particleup}{\particleup} = \Einterior{\particleup} + \springforceprefix \rdiffsquared{\particleup}{P}{\particleup}{1}$ \;
	\For{$\particledown = \particleup-1 \ldots 1$} {
		$\begin{aligned}
		\label{ln:enk-diff}
		\Efromto{\particledown}{\particleup} = \Efromto{\particledown+1}{\particleup}
                            &- \springforceprefix \rdiffsquared{\particleup}{P}{\particledown+1}{1}
                            \\
                            & + \springforceprefix \rdiffsquared{\particledown}{P}{\particledown+1}{1}
                            + \Einterior{\particledown}
                            \\
                            &+ \springforceprefix \rdiffsquared{\particleup}{P}{\particledown}{1}
		\label{ln:cycle-energies-end}
		\end{aligned}$
		\;
	}
}
$\label{ln:forward-potentials-start}$\tcp{compute potentials $\Vto{\particleup}$}
$\Vto{0} = 0$ \;
\For{$\particleup = 1 \ldots N$} {
	$\begin{aligned}
	\boltzmann{\Vto{\particleup}} = \frac{1}{\particleup} \sum_{k=1}^{\particleup}{\boltzmann{\left(\Vto{\particleup-k} + \Enk{\particleup}{k}\right)}} \label{ln:forward-potential-recursion}
	\label{ln:forward-potentials-end} \label{ln:potential-end}
	\end{aligned}$ \; 
}
$\label{ln:backward-potentials-start} \label{ln:forces-start}$
\tcp{compute potentials $\Vfrom{\particledown}$}
$\Vfrom{N+1} = 0$ \;
\For{$\particledown = N \ldots 1$} {
	$\begin{aligned}
	\boltzmann{\Vfrom{\particledown}} = \sum_{\ell=\particledown}^{N}{\frac{1}{\ell} \boltzmann{\left(\Efromto{\particledown}{\ell} + \Vfrom{\ell+1}\right)}} \label{ln:backward-potentials-recursion}
	\end{aligned}
	\label{ln:backward-potentials-end}$ \;
}
\tcp{compute connection probabilities}
$\label{ln:connection-probabilities-start}$ $\label{ln:direct-link-probability}$
\For{$\ell = 1 \ldots N-1$}{
	$\begin{aligned}
	\Prrepnext{\ell}{\ell+1} = 1 - \frac{1}{\boltzmann{\Vall}} \boltzmann{\left(\Vto{\ell} + \Vfrom{\ell + 1}\right)}
	\end{aligned}$ \;
}
$\label{ln:close-cycle-probability}$
\For{$\ell' = 1 \ldots N$} {
	\For{$\ell = \ell' \ldots N$} {
	    $\begin{aligned}
	    \Prrepnext{\ell}{\ell'} = \frac{1}{\ell} \frac{1}{\boltzmann{\Vall}} {\boltzmann{\left(\Vto{\ell'-1} + \Efromto{\ell'}{\ell} + \Vfrom{\ell+1}\right)}}
	    \label{ln:connection-probabilities-end}
	    \end{aligned}$ \;
	}
}
\tcp{compute force on interior beads}
$\label{ln:interior-force-start}$
\For{$\ell = 1 \ldots N$} {
	\For{$j = 2 \ldots P-1$} {
		$\begin{aligned}
		\beadforce{\ell}{j}{\Vall} = -\mass \springconstant^2 (2\posbead{\ell}{j} - \posbead{\ell}{j+1} - \posbead{\ell}{j-1})
		\end{aligned}$ %
		$\label{ln:forces-computation-end}$ $\label{ln:interior-force-end} \label{ln:forces-end}$
	}
} 
\tcp{compute force on exterior beads}
$\label{ln:forces-computation-start}$ $\label{ln:exterior-force-start}$
\For{$\ell = 1 \ldots N$} {
	$\begin{aligned}
	\beadforce{\ell}{1}{\Vall} &= -\sum_{\ell'=\ell-1}^{N}{\Prrepnext{\ell'}{\ell} \cdot \mass \springconstant^2 (2\posbead{\ell}{1} - \posbead{\ell}{2} - \posbead{\ell'}{P})}
	\\
	\beadforce{\ell}{P}{\Vall} &= -\sum_{\ell'=1}^{\ell+1}{\Prrepnext{\ell}{\ell'} \cdot \mass \springconstant^2 (2\posbead{\ell}{P} - \posbead{\ell'}{1} - \posbead{\ell}{P-1})}
	\end{aligned}
	\label{ln:exterior-force-end}$ \;
}
\caption{Fast PIMD-B\label{alg:full-pseudocode}}
\end{algorithm} In this section, we provide the full pseudocode of our algorithm that was presented in the main text.
The code appears in~\Cref{alg:full-pseudocode}.
The algorithm first computes the potential $\Vall$ (\crefrange{ln:potential-start}{ln:potential-end}), and then the forces $\beadforce{\ell}{j}{\Vall}$ (\crefrange{ln:forces-start}{ln:forces-end}).

For the potential, the algorithm first sums the spring energies due to interior springs (\crefrange{ln:interior-energy-start}{ln:interior-energy-end}), which are then used to evaluate all cycle energies $\Efromto{\particledown}{\particleup}$ through a recurrence relation on $\particledown$ (\crefrange{ln:cycle-energies-start}{ln:cycle-energies-end}).
A recurrence relation on $\Vto{1},\Vto{2},\ldots,\Vto{N-1},\Vto{N}$, which uses the cycle energies, yields the potential (\crefrange{ln:forward-potentials-start}{ln:forward-potentials-end}).

For the forces, the algorithm first evaluates the potentials $\Vfrom{1},\Vfrom{2},\ldots,\Vfrom{N-1},\Vfrom{N}$ through another recurrence relation (\crefrange{ln:backward-potentials-start}{ln:backward-potentials-end}). The connection probabilities are computed based on all aforementioned potentials and cycle energies (\crefrange{ln:connection-probabilities-start}{ln:connection-probabilities-end}). These are used to evaluate all the forces, separately on beads $1,P$ of all particles (\crefrange{ln:exterior-force-start}{ln:exterior-force-end}), and on the interior beads $j=2,\ldots,P-1$ of all particles (\crefrange{ln:interior-force-start}{ln:interior-force-end}). The expressions for the force on beads $1,P$ are slightly different from the expressions in the main text. There, the sum is over all potential neighbors $\ell=1,\ldots,N$, whereas \cref{ln:exterior-force-start} have different limits to the sum, excluding terms for connection probabilities that vanish, $\Prrepnext{\ell}{\ell'} = 0$ for $\ell' > \ell + 1$.

\begin{algorithm}
\DontPrintSemicolon
$\label{ln:original:potential-start} $\tcp{compute cycle energies $\Enk{\particledown}{\particleup}$}
\For{$\particleup = 1 \ldots N$} { $\label{ln:original:cycle-energies-start}$
	\For{$\particledown = 1 \ldots \particleup$} {
		$\begin{aligned}
		\Efromto{\particledown}{\particleup} &= 
								\springforceprefix \sum_{\ell=N-k+1}^{N}{\sum_{j=1}^{P-1}{\rdiffsquared{\ell}{j+1}{\ell}{j}}}
								\\
								&+
								\springforceprefix \sum_{\ell=N-k+1}^{N-1}{\rdiffsquared{\ell+1}{1}{\ell}{P}}
								+
								\springforceprefix \rdiffsquared{N-k+1}{1}{N}{P}
		\label{ln:original:cycle-energies-end}
		\end{aligned}$ 
		\;
	}
}
$\label{ln:original:forward-potentials-start}$\tcp{compute potentials $\Vto{\particleup}$}
$\Vto{0} = 0$ \;
\For{$\particleup = 1 \ldots N$} {
	$\begin{aligned}
	\boltzmann{\Vto{\particleup}} = \frac{1}{\particleup} \sum_{k=1}^{\particleup}{\boltzmann{\left(\Vto{\particleup-k} + \Enk{\particleup}{k}\right)}}
	\label{ln:original:forward-potentials-end}
	\end{aligned}$ \; \label{ln:original:potential-end}
}
\tcp{compute forces} $\label{ln:original:forces-computation-start}$
\For{$\ell = 1 \ldots N$} { $\label{ln:original:forces-cycle-energy-derivative-start}$
	\For{$j = 1 \ldots P$} {
		$\begin{aligned}
		&\beadderive{\ell}{j}{\Enk{\particleup}{k}} = \mass \springconstant^2 (2\posbead{\ell}{j} - \posbead{\ell}{j+1} - \posbead{\ell}{j-1}), 
		\\
		&\mbox{where $\posbead{\ell}{0} = \posbead{\ell-1}{P}, \, \posbead{\ell}{P+1} = \posbead{\ell+1}{1}$ except $\posbead{N-k+1}{0} = \posbead{N}{P}, \, \posbead{N}{P+1} = \posbead{N-k+1}{1}$.}
		\end{aligned}$ \; $\label{ln:original:forces-cycle-energy-derivative-end}$
	}
}
\For{$\ell = 1 \ldots N$} { $\label{ln:original:forces-recurrence-start}$
	\For{$j = 1 \ldots P$} {
		\For{$\particleup = 1 \ldots N$} {
			$\begin{aligned}
			\beadforce{\ell}{j}{\Vto{\particleup}} = 
			-\frac{
				\sum_{k=1}^{\particleup}{
					\left(\beadderive{\ell}{j}{\Enk{\particleup}{k}} + \beadderive{\ell}{j}{\Vto{\particleup - k}}\right) 
					\boltzmann{\left(\Enk{\particleup}{k} + \Vto{\particleup}\right)}
				}
			}
			{\particleup \cdot \boltzmann{\Vto{\particleup}}}
			\end{aligned}$ %
			$\label{ln:original:forces-recurrence-end}$ $\label{ln:original:forces-computation-end}$
		}
	}
} 
\caption{Original PIMD-B\label{alg:original-full-pseudocode}}
\end{algorithm} For comparison, a pseudocode for the original method~\cite{hirshberg2019path} is shown in~\Cref{alg:original-full-pseudocode}.
It has the same overall structure: the potential is computed (\crefrange{ln:original:potential-start}{ln:original:potential-end}) by evaluating the cycle energies (\crefrange{ln:original:cycle-energies-start}{ln:original:cycle-energies-end}) and using them in the recurrence relation (\crefrange{ln:original:forward-potentials-start}{ln:original:forward-potentials-end}). However, the evaluation of the cycle energies is less efficient in this algorithm, and scales as $\bigO(PN^3)$ rather than $\bigO(N^2 + PN)$.
The original LAMMPS implementation did not separate the evaluation of the cycle energies from the recurrence relation. This has no effect on the asymptotic complexity, but may change the constant factor.
Then, the forces are evaluated by a recurrence relation (\crefrange{ln:original:forces-computation-start}{ln:original:forces-computation-end}), first evaluating spatial derivatives of cycle energies (\crefrange{ln:original:forces-cycle-energy-derivative-start}{ln:original:forces-cycle-energy-derivative-end}) and then using a recurrence relation for the derivatives of the potential (\crefrange{ln:original:forces-recurrence-start}{ln:original:forces-recurrence-end}). The latter step takes $\bigO(N^2)$ time for every bead, and $\bigO(PN^3)$ overall, worse than in the new algorithm.

\subsection{Numerical Stability}
\label{sec:numerical-stability}
Computing the potentials $\Vto{\particleup}$ (\crefrange{ln:forward-potentials-start}{ln:forward-potentials-end} in~\Cref{alg:full-pseudocode}) and $\Vfrom{\particledown}$ (\crefrange{ln:backward-potentials-start}{ln:backward-potentials-end}) involve a sum of exponentials, each of which may be too large to represent in ordinary floating-point representations. As described in the SI of Ref.~\citenum{hirshberg2019path}, the trick is to shift all exponents by constant amount $\tilde{E}$. \Cref{ln:forward-potential-recursion} of~\Cref{alg:full-pseudocode} is replaced by
\begin{align*}
	\Vto{\particleup} &= \tilde{E}_1 - \frac{1}{\beta}\ln{\left(\frac{1}{\particleup} \sum_{k=1}^{\particleup}{\boltzmann{\left(\Vto{\particleup-k} + \Enk{\particleup}{k} - \tilde{E}_1\right)}}\right)}, \\
\intertext{and~\cref{ln:backward-potentials-recursion} by}
	\Vfrom{\particledown} &= \tilde{E}_2 - \frac{1}{\beta} \ln{\left(\sum_{\ell=\particledown}^{N}{\frac{1}{\ell} \boltzmann{\left(\Efromto{\particledown}{\ell} + \Vfrom{\ell+1} - \tilde{E}_2\right)}}\right)},
\end{align*}
where every choice of $\tilde{E}_1,\tilde{E}_2$ is valid.
Following~\citet{PhysRevE.106.025309}, we used the value leading to the largest exponent in each recursive step:
\begin{align}
	\tilde{E}_1 &= \min_{k=1,\ldots,\particleup}{\left(\Vto{\particleup-k} + \Enk{\particleup}{k}\right)}, \label{eq:full-elong}
	\\
	\tilde{E}_2 &= \min_{\ell=\particledown,\ldots,N}{\left(\Efromto{\particledown}{\ell} + \Vfrom{\ell+1}\right)}.
\end{align}
We used a simpler version for non-interacting and small interacting trapped systems,
\begin{align}
	\tilde{E}_1 &= \min{\left(\Efromto{1}{\particleup}, \Vto{\particleup-1} + \Efromto{\particleup}{\particleup}\right)}, \label{eq:original-elong}
	\\
	\tilde{E}_2 &= \min{\left(\Efromto{\particledown}{\particledown} + \Vfrom{\particledown+1}, \Efromto{\particledown}{N}\right)},
\end{align}
which is similar to the original implementation in LAMMPS. However, this was inadequate for larger systems or in the presence of strong interaction.
\section{Additional Computational Details}
\label{sec:benchmarks-supplementary}

All simulations were performed on a cluster of servers, each with two Intel Xeon Platinum 9242 CPU @ 2.30GHz,  386GB RAM, and a total of 96 cores. $P$ cores were used for each simulation in LAMMPS (based on its replica mechanism built on top of OpenMPI). A single core for each simulation in i-PI was used. A whole server was blocked out for performance measurements.

\subsection{Scalability of the Implementation in i-Pi}
\label{sec:i-pi-scalability}

\begin{figure}[t]
    \includegraphics[width=0.47\textwidth]{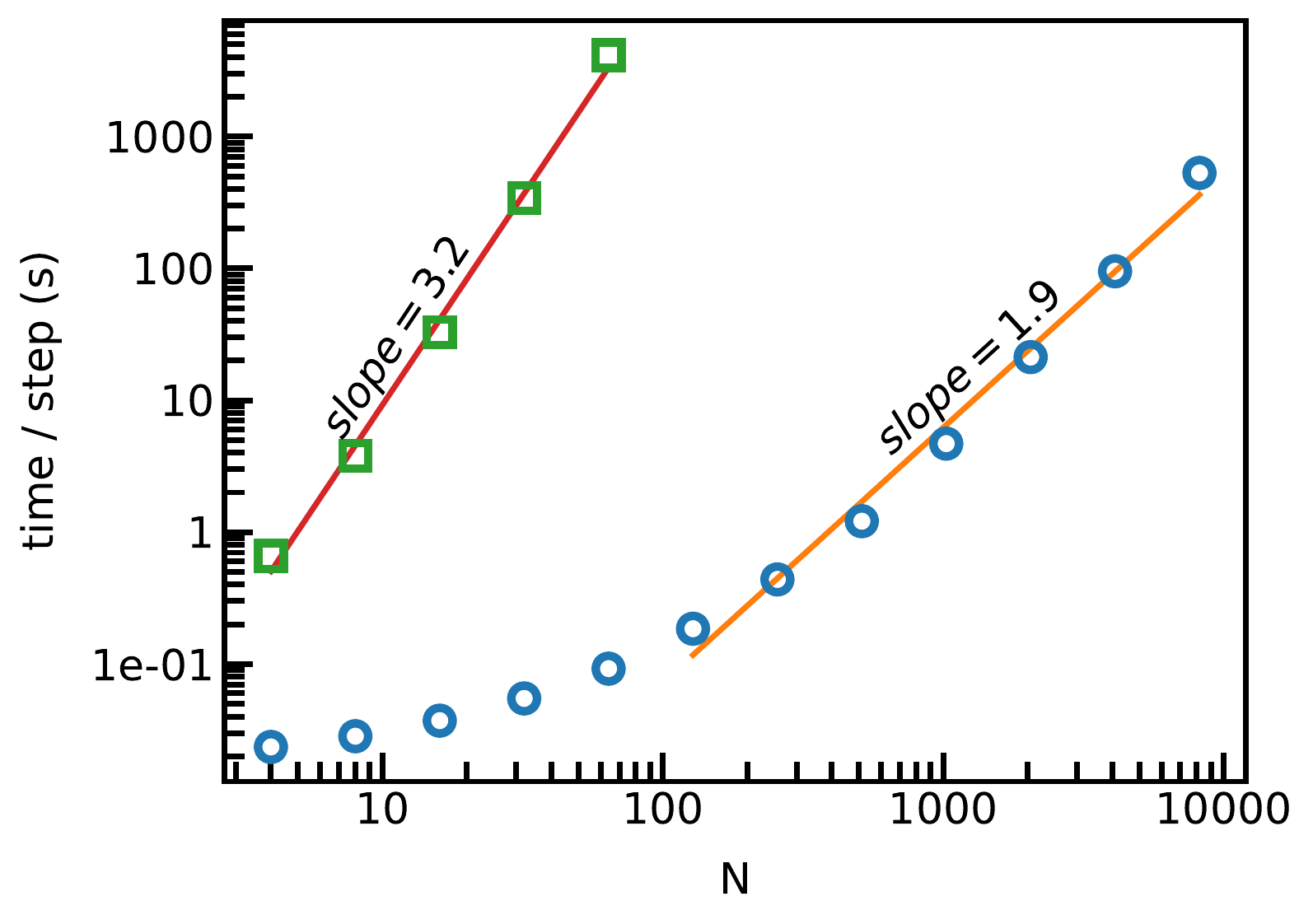}%
\caption{
\label{fig:i-pi-speedup}
  The scaling in i-PI of the original and new methods as with the number of particles for a system of $N$ non-interacting cold trapped bosons on a log-log scale. Blue circles, green squares represent the new and original algorithms, respectively.
}
\end{figure}

 \Cref{fig:i-pi-speedup} presents the scalability of the implementation of the algorithm in i-Pi~\cite{KAPIL2019214}, showing the time per simulation step as a function of the number of particles $N$ in log-log scale, for a system of trapped cold non-interacting bosons with $P=32$. The time measured is of the code that calculates the exchange potential and forces only, excluding the external force calculation and the propagator. We see that in i-Pi, the new algorithm indeed scales better than the original algorithm. The slope of the linear fit from sufficiently large $N$ %
is $1.9$ for our new algorithm, compared to $3.2$ for the original algorithm in the i-PI implementation.

\subsection{Numerical Validation}
\begin{figure}[t]
    \subfloat{\includegraphics[width=0.47\textwidth]{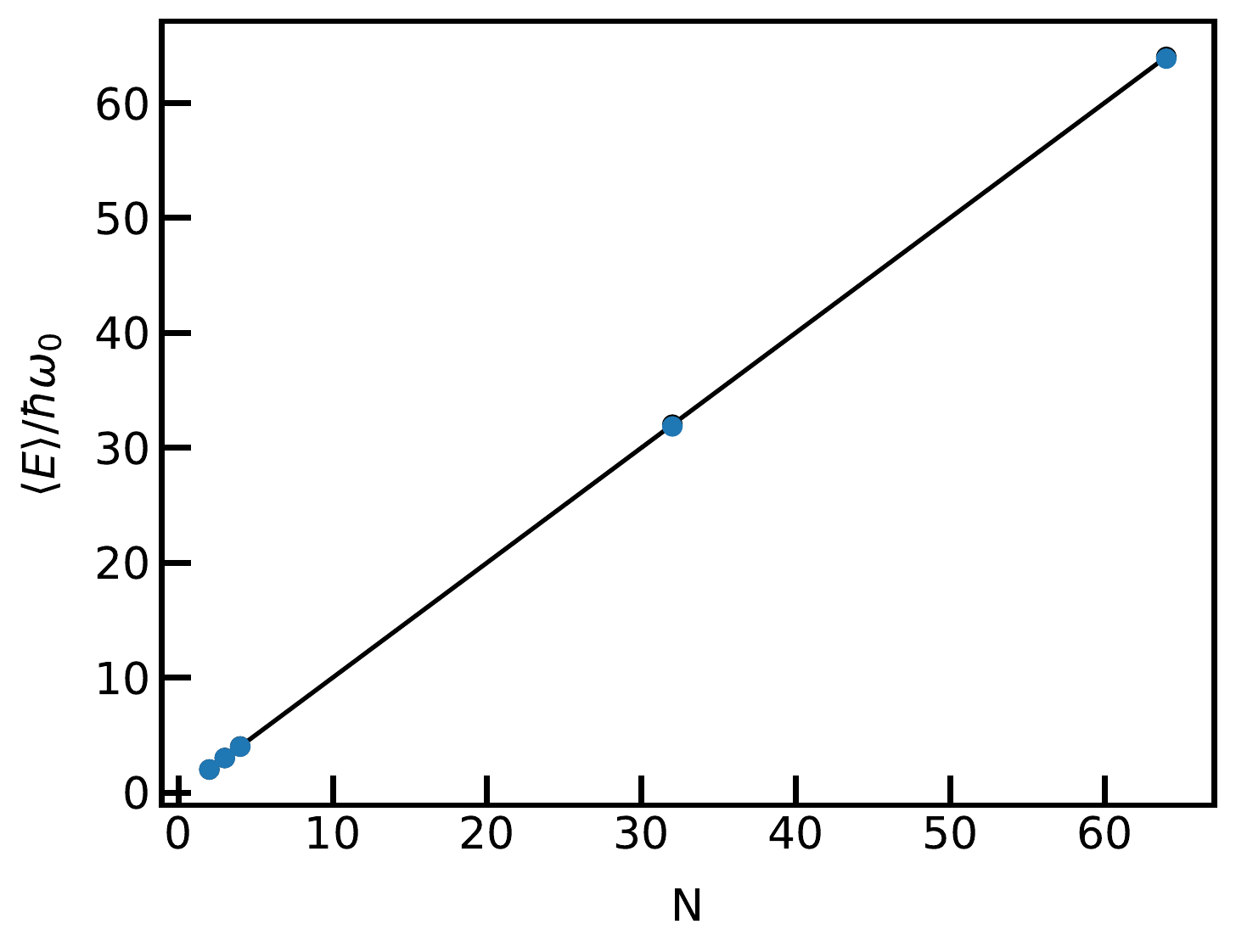}}%
\hspace{0.1in}
    \subfloat{\includegraphics[width=0.47\textwidth]{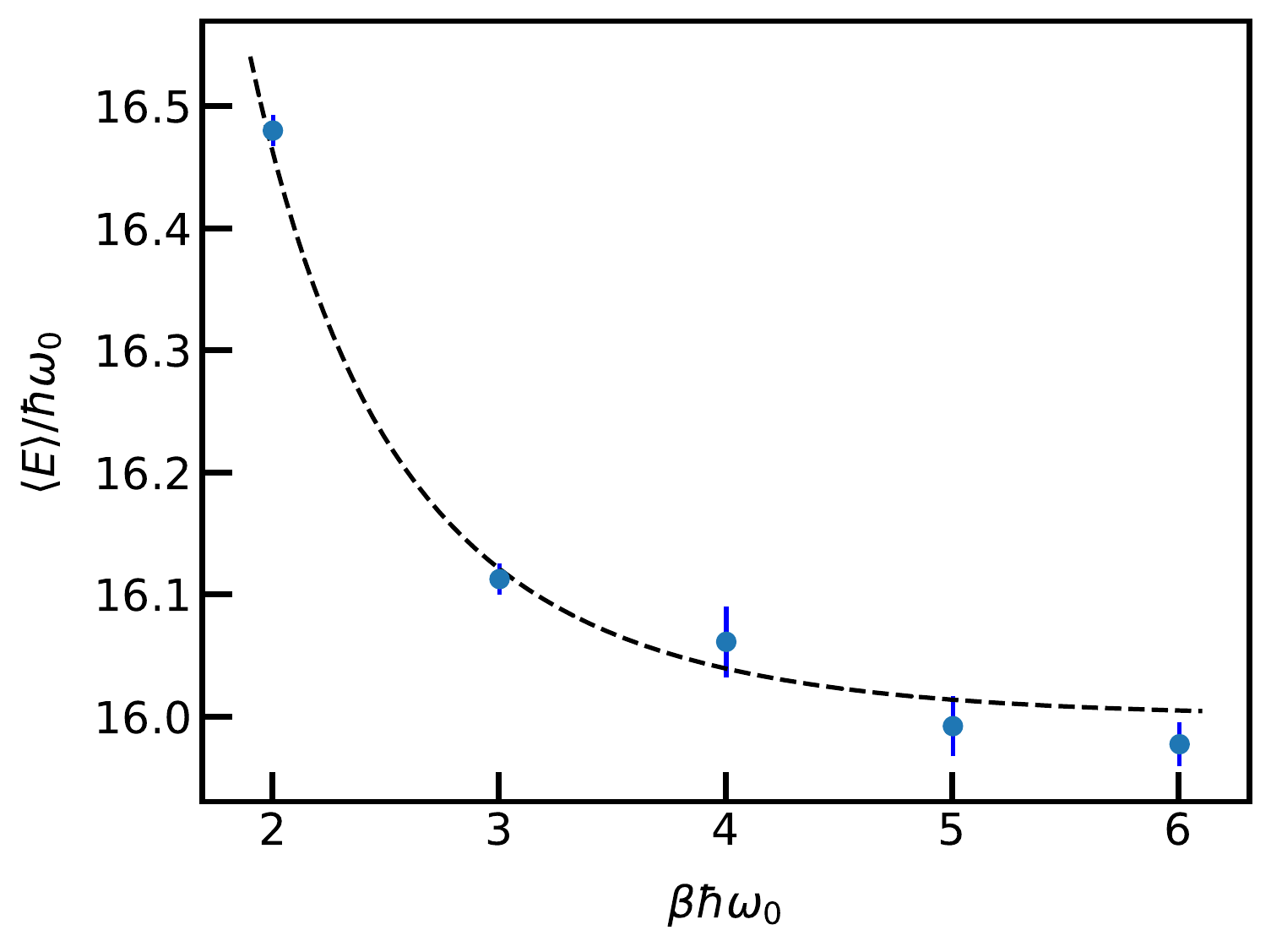}}%
\caption{
  \label{fig:harmonic-energy-experiments}
  The energy of non-interacting particles in a harmonic trap as a function of $N$ with fixed temperature corresponding to $\beta \hbar \omega =6$ (left), and as function of the temperature with fixed $N=16$ (right), and their agreement with the analytical result. $P=36$.
}
\end{figure} %
\begin{figure}[t]
    \subfloat{\includegraphics[width=0.47\textwidth]{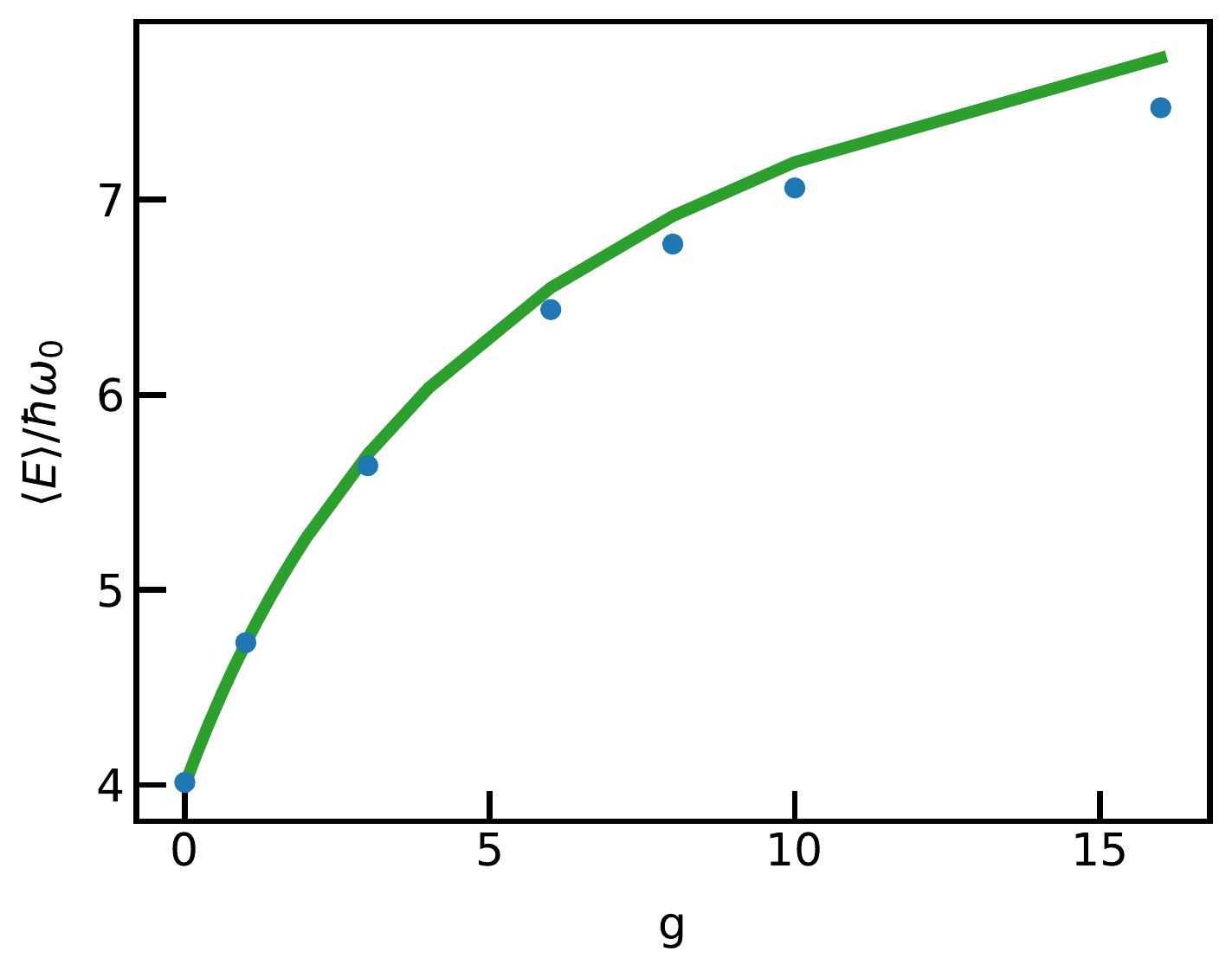}}%
\hspace{0.1in}
    \subfloat{\includegraphics[width=0.47\textwidth]{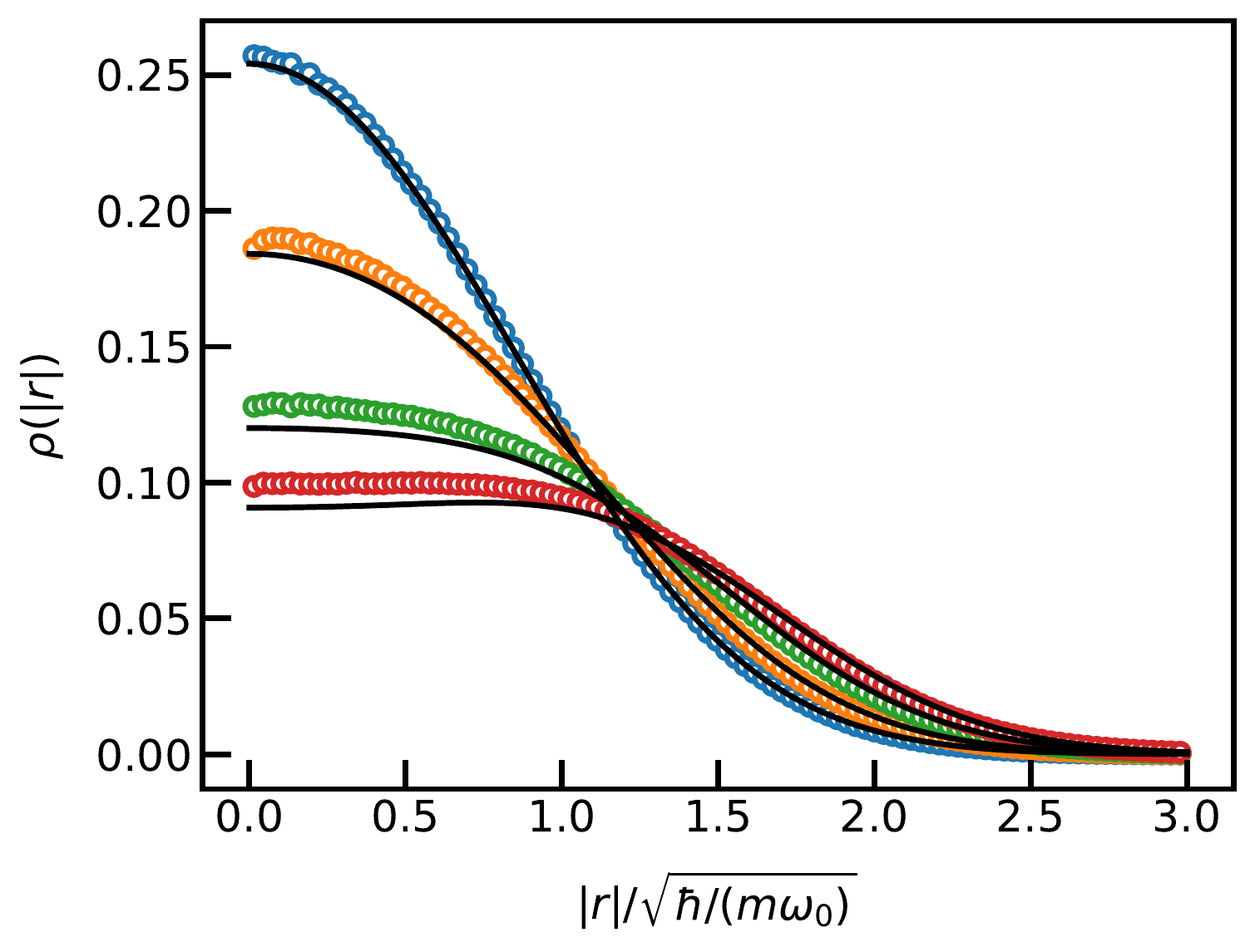}}%
\caption{
  \label{fig:gaussian-experiments}
  The energy and radial density of $N=4$ interacting particles with Gaussian repulsion of strength $g$ (in harmonic-oscillator units) in a harmonic trap and their agreement with previous results from numerically diagonalizing the Hamiltonian~\cite{Mujal2017}.
}
\end{figure} 
In this section, we present simulation results validating the correctness of our improved algorithm.
\Cref{fig:harmonic-energy-experiments} shows the energy of trapped bosons in a harmonic trap:
1) as a function of $N$, at a temperature corresponding to $\hbar \beta \omega_0 = 6$, and 2) as a function of temperature, with $N=16$.
We use $P=36$, and a time step of $1$ fs. The simulations were performed for $10^6$ steps with the first 20\% discarded as equilibration.
Three independent runs were used; the error bars of one standard deviation are smaller than the marker size.
The results of the simulation are compared with the analytical result.
The maximal relative deviation is 0.5\%.
\Cref{fig:gaussian-experiments} shows results for trapped bosons interacting through a Gaussian repulsive interaction (see Ref.~\citenum{hirshberg2019path} for details).
We calculate the energy and radial density for $N=4$ bosons for different interaction strengths, controlled by the parameter $g$.
We use $P=72$, and a time step of $1$ fs. The simulations were performed for $10 \cdot 10^6$ steps with the first 20\% discarded as equilibration.
Three independent runs were used; the error bars of one standard deviation are smaller than the marker size.
The results are compared to the results of~\citet{Mujal2017}. The maximum relative deviation in the energy is 3.3\% and the mean deviation is 1.5\%, close to the deviations of~\citet{hirshberg2019path}, which were 2.9\% and 1.3\% respectively.  
The mean relative deviation in the radial density is at most $\mathord{\sim} 5.3\%$.

\subsection{Additional benchmarks and applications}
\begin{figure*}[t]
\center
    \subfloat{
    \includegraphics[width=0.42\textwidth]{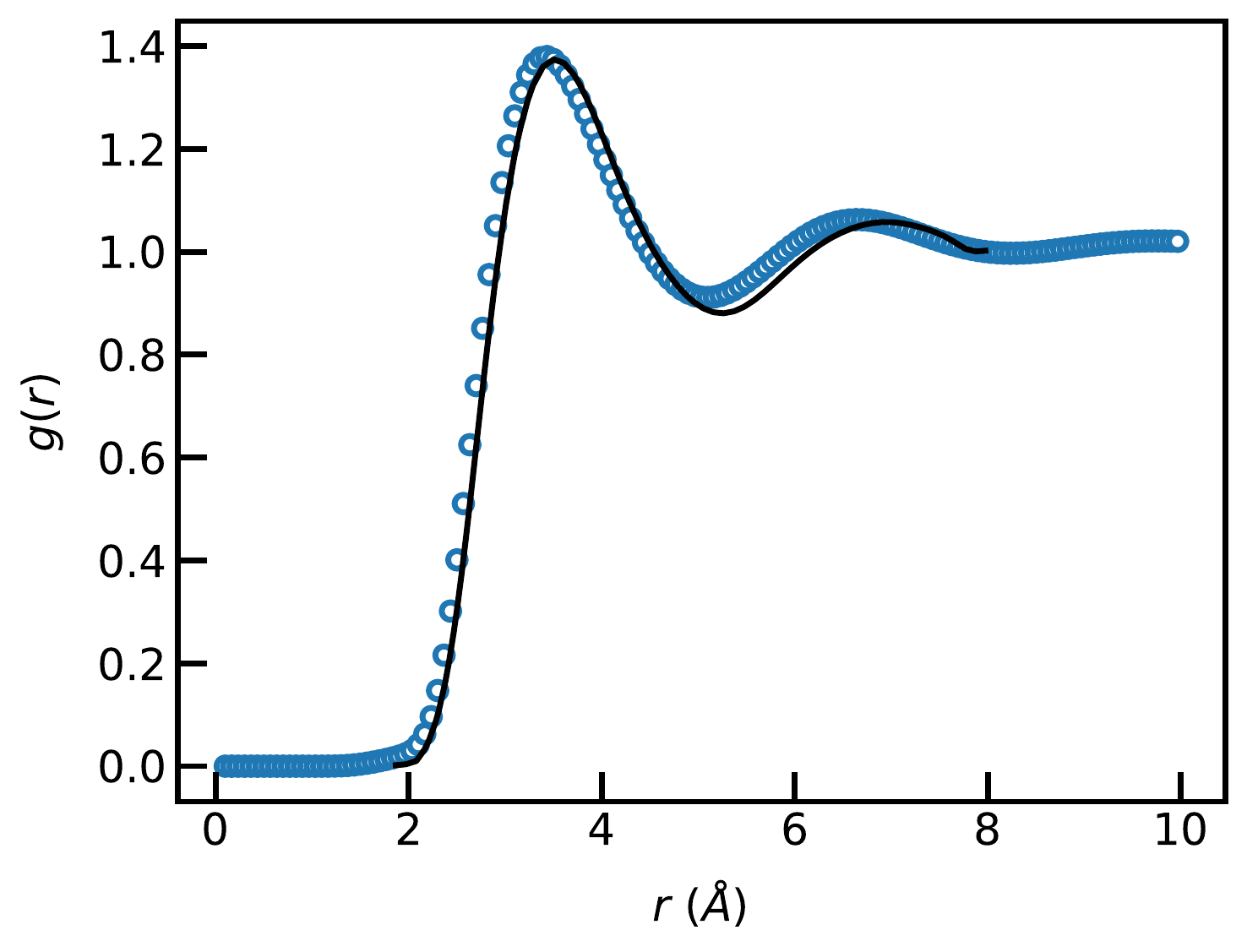}
    }%
\hspace{0.1in}
    \subfloat{\includegraphics[width=0.42\textwidth]{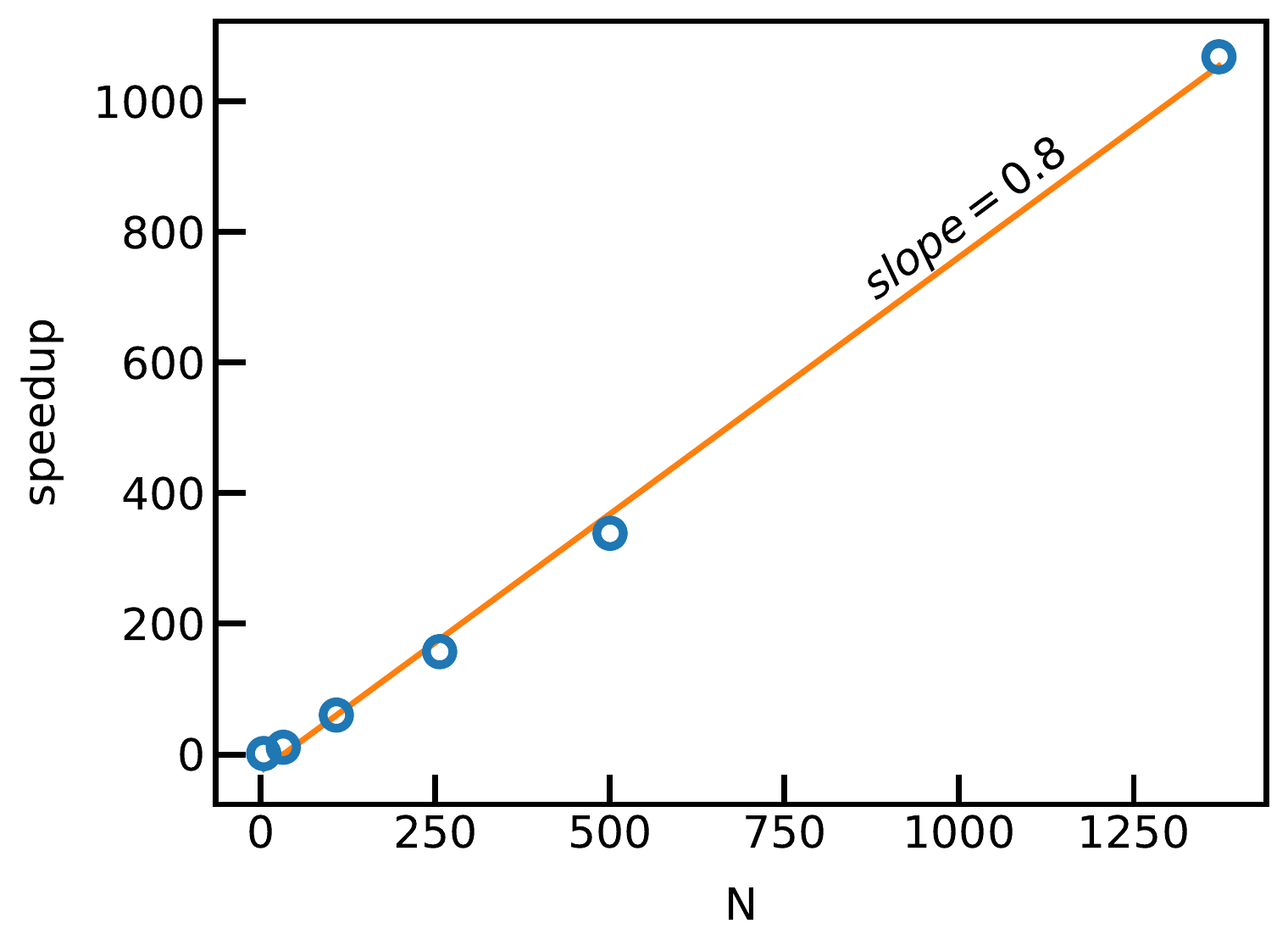}}%
\caption{
\label{fig:liquid-he-supplementary}
  Simulations of superfluid liquid \ce{^{4}He}. Left: the pair correlation of $N=1372$ atoms at temperature $1.379$K with and a density of $0.02182 \AA^{-3}$, compared with the PIMC results of~\citet{ceperley1995path}. %
  Right: the speedup of the improved algorithm over the original method as a function of the number of \ce{^{4}He} molecules, with $P=64$ beads, computed at $1.21$K and the same density. %
}
\end{figure*} %
\begin{figure*}[t]
\center
    \subfloat{
    \includegraphics[width=0.5\textwidth]{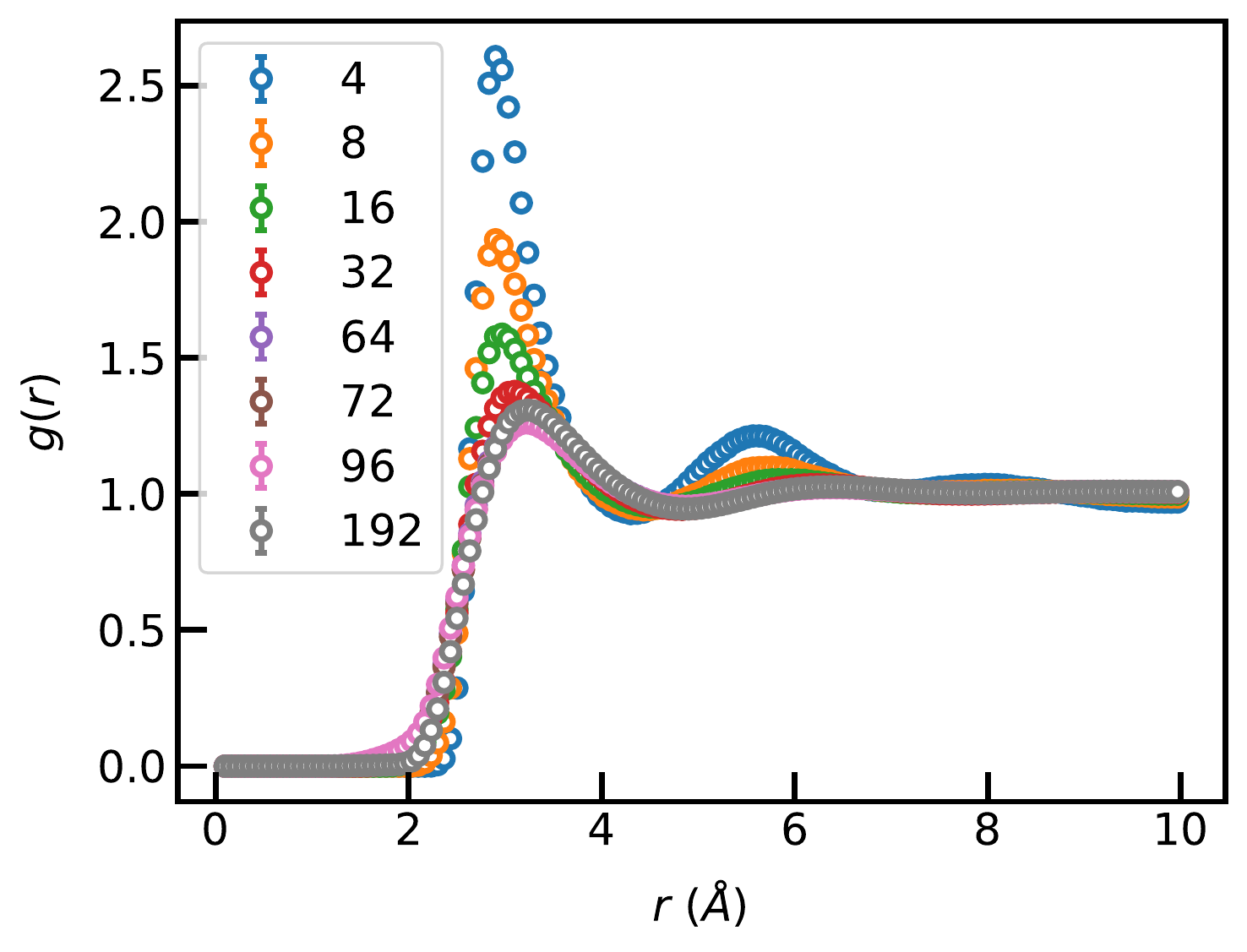}
    }%
\caption{
\label{fig:liquid-he-bead-convergence}
  Bead convergence of the pair correlation for $N=108$ \ce{^{4}He} atoms, with different values of $P$. The temperature is $1.21$K and the density is $0.02182 \AA^{-3}$.
}
\end{figure*} \Cref{fig:liquid-he-supplementary} compares the pair correlation function for superfluid liquid \ce{^{4}He} to the pair correlation from a PIMC simulation of~\cite{ceperley1995path} (see Fig.\ 16 there).  %
The number of beads was chosen based on a convergence graph for $N=108$ atoms at $1.21$K (see~\Cref{fig:liquid-he-bead-convergence}).

The speedup of the new algorithm over the original one in a system of superfluid liquid atoms is displayed in~\Cref{fig:liquid-he-supplementary}.
The simulation speedup as a function of the number of atoms $N$ 
 at a temperature of $1.21$K and a density of $0.02182 \AA^{-3}$
 with $P=64$ is %
 approximately linear ($R^2=0.997$) with a slope of $\mathord{\sim} 0.8$. We attribute the reduced speedup factor compared to the case of trapped bosons (displayed in the main text) to the cost of evaluating particle-particle interactions, which is not sped up by the new algorithm. %

\begin{figure*}
\center
    \subfloat{\includegraphics[width=0.35\textwidth]{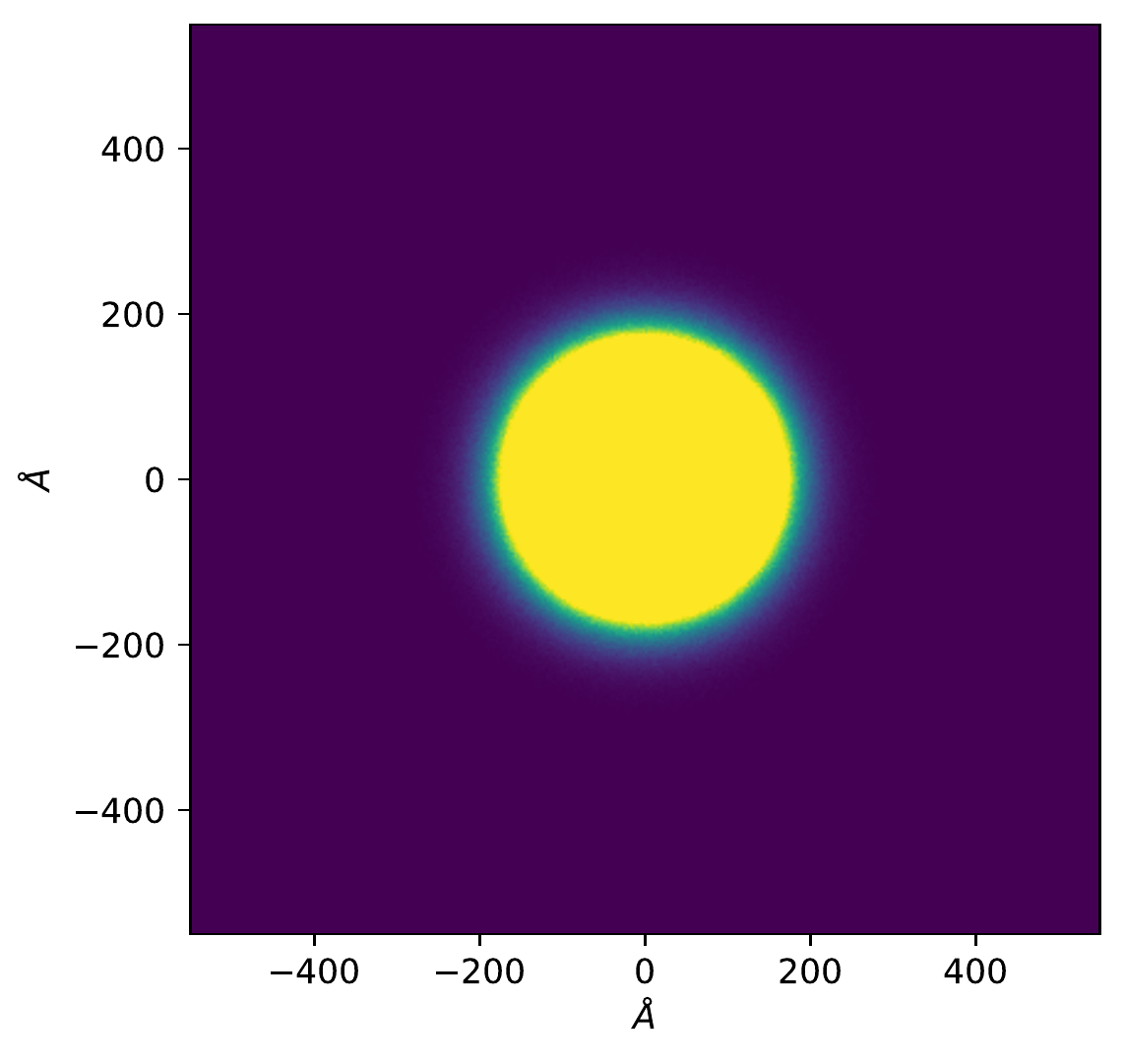}}%
\hspace{0.1in}
    \subfloat{\includegraphics[width=0.35\textwidth]{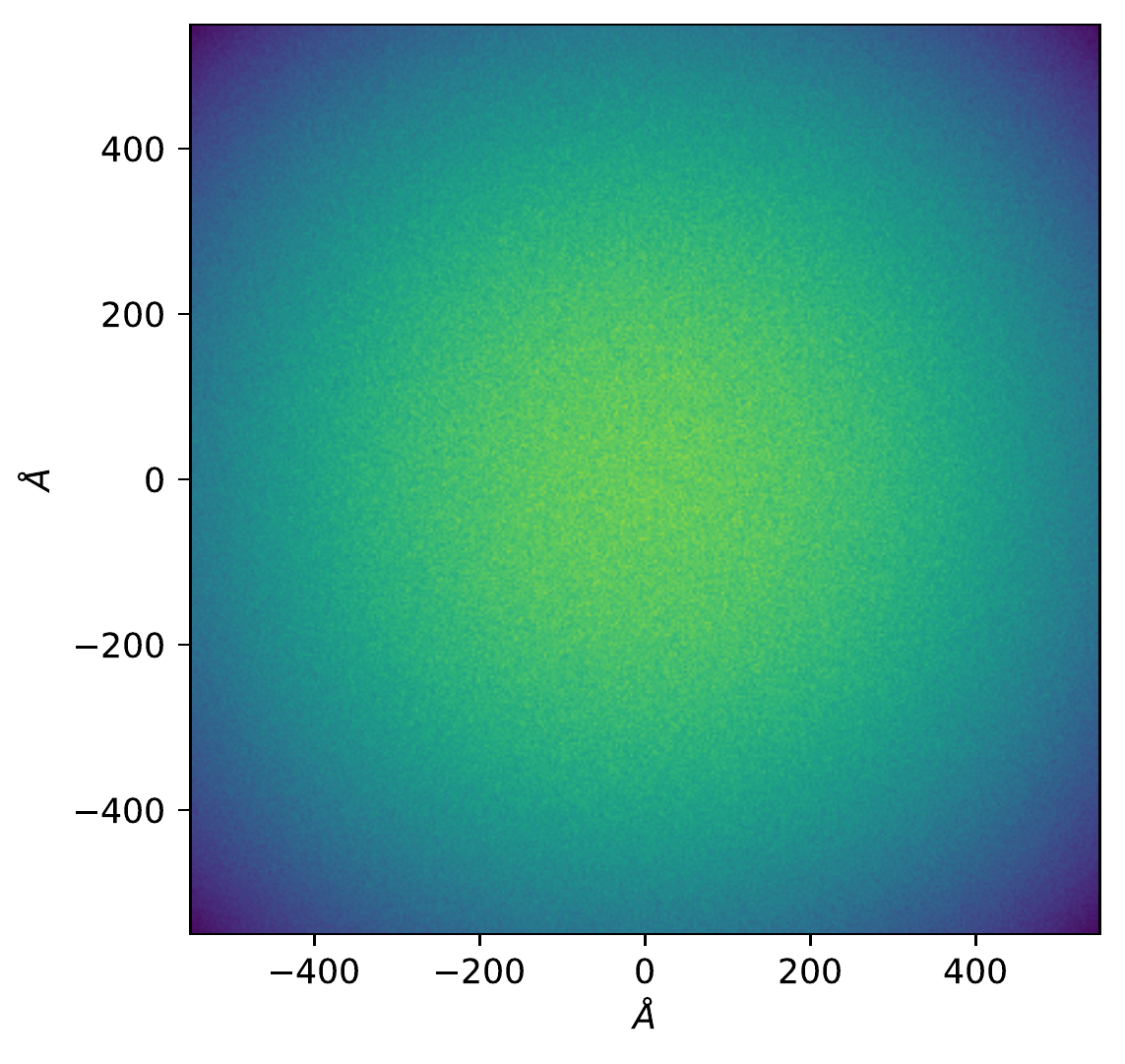}}%
\caption{
  \label{fig:2d-density}
  Two-dimensional density of $N=1600$ bosons in a harmonic trap without interaction (left) and with Gaussian repulsive interaction $g=3$ (right). The simulations last $\leq 9$ days, instead of $> 20$ years with the previous algorithm.
}
\end{figure*} \Cref{fig:2d-density} shows the result of simulations of $N=1600$ trapped cold bosons, with and without Gaussian repulsive interaction. (For comparison, the largest trapped system studied in previous bosonic PIMD methods was for $N=64$~\cite{hirshberg2019path}.)
We use $P=36$, appropriate for the temperature of $11.6$K ($\beta \hbar \omega=3$) and a time step of $1$ fs. The simulations were performed for $3 \cdot 10^6$ steps with the first 20\% discarded as equilibration. The repulsive potential is the same as described in ref.~\citenum{hirshberg2019path}, with $g=3$. 
As expected, the repulsive interaction between the particles pushes them to spread in the trap resulting in a wider two-dimensional density.

\section{Cycle Notation of Permutations}
\label{sec:cycle-notation}
The cycle notation is a common way to identify permutations.
The cycle notation of a permutation $\sigma$ is constructed in the following way: begin with particle $1$, then add to the cycle the particles obtained by applying $\sigma$ successively, until $\cyclenotate{1, \sigma(1), \sigma(\sigma(1)), \ldots}$ returns to particle 1. The properties of the permutation guarantee that the cycle will close. Then, to construct the other cycles of $\sigma$, repeat the process with the particle with smallest index that was not yet included in previous cycles, until all particles are assigned. The set of cycles constructed in this way is the cycle notation of $\sigma$. For example, the cycle notation for the permutation $\permutationthree{2}{1}{3}$ is $\cyclenotate{12}\cyclenotate{3}$. 

The cycle notation is especially appealing in our context, because 
the ring polymer configuration is identified from the cycle notation in the following way:
each cycle of $\sigma$ translates to a separate ring polymer in $\Eperm{\sigma}$ in which the corresponding particles are connected in the same order. %
\section{Some Combinatorics of Permutations}
\label{sec:combinatorics}

In this section, we detail some of the preceding claims on the combinatorics of permutations satisfying certain properties.
\begin{itemize}
	\item The number of permutations on $N$ elements is $\fact{N}$, which asymptotically is $\Theta((N/e)^N)$ (by Stirling's approximation), growing (faster than) exponentially.

	\item The number of representative permutations $\rep{\sigma}$ over $N$ elements is $2^{N-1}$, which is exponential. 
	To see this, imagine the process of constructing $\rep{\sigma}$ as scanning the interval $[1,N]$ from left to right, and simultaneously the cycle notation of the permutation, and ``cutting'' a new interval whenever a cycle ends.
	There are $N-1$ places in which $\rep{\sigma}$ may cut, and every combination of the choices---whether to cut after $1$, after $2$, and so forth---yields a different representative permutation. There are $2$ choices in each location---whether to cut or not---so this process can generate $2^{N-1}$ representative permutations.

    \item In the main text, when discussing how to compute connection probabilities, we argue that the number of permutations $\sigma$ on $N$ elements in which $\ell$ is connected to $\ell'$ grows exponentially. The number is $\fact{N}/N = \pfact{N-1}$ (from symmetry considerations), which indeed grows (faster than) exponentially. %

	\item The number of permutations with unique topologies (unique cycle type) is the number of partitions of the number $N$, i.e., the number of unique ways to write $N$ as a sum of positive integers, which is known~\cite{10.1112/plms/s2-17.1.75} to be exponential in $\sqrt{N}$.
\end{itemize}
\section{Proof of the Force Evaluation Algorithm}
\label{sec:proofs}
In this section, we present in detail the central derivations leading to the improved algorithm for evaluating the forces.
We first define the ring polymer potential using a sum over representative permutations, and show that it coincides with the recursive definition of the potential from the original method (\Cref{sec:proof-sum-of-representatives}).
From this we then derive a probabilistic expression for evaluating the forces (\Cref{sec:proof-force-as-expected-force}).
Finally, we derive expressions for the connection probabilities (\Cref{sec:connection-probabilities-proofs}) on which the efficient calculation of the forces is based.

\paragraph{Notation}
We begin with some notation.
We denote the set of permutations over $[1,N] = \set{1,\ldots,N}$ by $\Symfromto{1}{N}$. 
A \emph{cycle} is a list of elements. For example, $(N-k+1,\ldots,N)$ denotes the cycle over these elements. We sometimes write a permutation $\sigma$ as a list of (disjoint) cycles, e.g., $\sigma = c_1 \cdot \ldots \cdot c_m$ (the product means adding another cycle).
In a permutation $\sigma$, the cycle to which an element $i$ belongs is denoted $\cycleof{\sigma}{i}$.
Recall from the main text that the ring polymer energy of a permutation $\sigma$ is
\begin{equation*}
    \Eperm{\sigma} = \springforceprefix \sum_{\ell=1}^{N}{ \sum_{j=1}^{P}{\rdiffsquared{\ell}{j}{\ell}{j+1}}}, \mbox{ where $\beadpos{\ell}{P+1} = \beadpos{\sigma(\ell)}{1}$.}
\end{equation*}

\subsection{Sum Over Representative Permutations}
\label{sec:proof-sum-of-representatives}

Our alternative form for the potential uses the concept of \emph{representative permutations}:
\begin{definition}
\label{def:representative-permutation}
Given a permutation $\sigma$ over $N$ elements, the \emph{representative permutation} $\rep{\sigma} \in \SymN{N}$ is the permutation obtained from the cycle notation of $\sigma$ as follows (and as explained in the main text):
Start from the canonical cycle notation of $\sigma$, which crucially sorts the cycles of $\sigma$ in increasing order in their highest element, $\sigma=c_1 \ldots c_m$, where $\max{c_i} < \max{c_{i+1}}$.
Construct $\rep{\sigma}$ by replacing the elements in the cycles with $1,\ldots,N$ in that order, that is, $\rep{\sigma}=c'_1 \ldots c'_m$ where
$$c'_{i+1} = (1 + \max{c'_i},\ldots,\card{c_{i+1}} + \max{c'_i})$$ with $\max{c'_0}=0$.
\end{definition}
Note that $\rep{\sigma}$ is indeed always a permutation of $1,\ldots,N$, provided that $\sigma$ is, and that it has the same cycle type, since $\card{c'_i} = \card{c_i}$.

Our expression for the bosonic ring polymer potential is as follows:
\begin{definition}
\label{def:forward-potential}
We define the potential $\Vto{N}$ by
\begin{equation*}
    \boltzmann{\Vto{N}} = \frac{1}{\fact{N}} \sum_{\sigma \in \SymN{N}}{\boltzmann{\Eperm{\rep{\sigma}}}}.
\end{equation*}
\end{definition}

Our first insight is that this definition captures exactly the bosonic ring polymer potential of~\citet{hirshberg2019path}, as the following theorem shows:
\begin{theorem}
\label{thm:forward-potential-recurrence}
The boson potentials $\Vto{1},\ldots,\Vto{N}$ defined in~\Cref{def:forward-potential} satisfy the recursion formula
\begin{equation*}
	\boltzmann{\Vto{N}} = \frac{1}{N} \sum_{k=1}^{N}{\boltzmann{\left(\Vto{N - k} + \Enk{N}{k}\right)}},
\end{equation*}
with the notation $\Vto{0} = 0$.
\end{theorem}
The intuition underlying the proof %
is that the recursive definition is a manifestation of a recursive structure in $\repsym$.
The last cycle in $\rep{\sigma}$ is of length $k$ and always includes $N$. The other cycles in $\rep{\sigma}$ repeat this process for the remaining $N-k$ particles.
The sum in~\Cref{def:forward-potential} corresponds to the different possible values of $k=1,...,N$, and the prefactor accounts for the fraction of permutations in which $N$ is in a cycle of length $k$.
The recursive structure of $\repsym$ and the ring polymer energy is formalized in the following lemma:
\begin{lemma}
\label{lem:representative-energy-recurrence}
Let $\sigma \in \SymN{N}$ be a permutation where the cycle to which $N$ belongs contains $k$ elements, that is, $\card{\cycleof{\sigma}{N}}=k$.
Consider the permutation $\sigma \setminus \cycleof{\sigma}{N}$, that is, the permutation over $N-k$ elements that is constructed from the cycles of $\sigma$ \emph{excluding} the cycle $\cycleof{\sigma}{N}$.
Then
\begin{equation*}
	\Eperm{\rep{\sigma}} = \Eperm{\rep{\sigma \setminus \cycleof{\sigma}{N}}} + \Enk{N}{k}.
\end{equation*}
\end{lemma}
\begin{proof}
In the canonical cycle notation of $\sigma$, the last cycle is always $\cycleof{\sigma}{N}$. Hence, in the representative permutation, according to~\Cref{def:representative-permutation}, this is transformed to 
$$\rep{\sigma} = \rep{\sigma \setminus \cycleof{\sigma}{N}} \cdot (N-k+1,\ldots,N).$$
The claim follows because $\Eperm{\rep{\sigma}}$ factors to a sum over the cycles of $\rep{\sigma}$.
\end{proof}

We are now ready to prove the theorem.
\begin{proof}[Proof of~\Cref{thm:forward-potential-recurrence}]
We split the sum over all permutations according the size of the cycle to which the element $N$ belongs.
\begin{align}
		\boltzmann{\Vto{N}} 
	&= 
		\frac{1}{\fact{N}} \sum_{\sigma \in \SymN{N}}{\boltzmann{\Eperm{\rep{\sigma}}}}
	\\
	&=
	\label{eq:first-cycle-argument-start}
		\frac{1}{\fact{N}} \sum_{k=1}^{N}{\sum_{\sumcond{\sigma \in \SymN{N}}{\card{\cycleof{\sigma}{N}}=k}}{\boltzmann{\Eperm{\rep{\sigma}}}}}
	\\
	\intertext{By~\Cref{lem:representative-energy-recurrence},}
	&=
		\frac{1}{\fact{N}} \sum_{k=1}^{N}{\sum_{\sumcond{\sigma \in \SymN{N}}{\card{\cycleof{\sigma}{N}}=k}}{\boltzmann{\left(
											\Eperm{\rep{\sigma \setminus \cycleof{\sigma}{N}}}
											+ \Enk{N}{k}
										 \right)}}}.
	\\
	\intertext{Since $\Enk{N}{k}$ is independent of $\sigma$,} %
	&=
		\frac{1}{\fact{N}} \sum_{k=1}^{N}{\boltzmann{\Enk{N}{k}} 
			\sum_{\sumcond{\sigma \in \SymN{N}}{\card{\cycleof{\sigma}{N}}=k}}{\boltzmann{\Eperm{\rep{\sigma \setminus \cycleof{\sigma}{N}}}}}}. \label{eq:proof-sum-all-step}
\end{align}

Consider the term on the right, the sum over all permutations over $N$ elements where the cycle of $N$ is of length $k$. We will transform it to use a sum over all permutations over $N-k$ elements only.
To this end, define a function $f$, that given a permutation $\tilde{\sigma}$ over $N-k$ possibly non-consecutive elements, returns a permutation of $[1,N-k]$, by renaming the elements in $\tilde{\sigma}$ to $1,\ldots,N-k$ while respecting the order.
It is easy to see that the representative permutation is unaltered by $f$: $\rep{\tilde{\sigma}} = \rep{f(\tilde{\sigma})}$, seeing that the renaming of the cycles that $f$ performs does not change the relative order between them in the canonical cycle notation.
Thus, it is valid to replace the representative of $\sigma'$ with the representative of $f(\sigma')$:
\begin{align}
	&\sum_{\sumcond{\sigma \in \SymN{N}}{\card{\cycleof{\sigma}{N}}=k}}{\boltzmann{\Eperm{\rep{\sigma \setminus \cycleof{\sigma}{N}}}}}
	\\
	&=\sum_{\sumcond{\sigma \in \SymN{N}}{\card{\cycleof{\sigma}{N}}=k}}{\boltzmann{\Eperm{\rep{f(\sigma \setminus \cycleof{\sigma}{N})}}}}
	\\
	\intertext{We now introduce a sum over permutations over $[1,N-k]$. At first, this sum is trivial, over only one element, because $f$ is a function:}
	&=
			\sum_{\sumcond{\sigma \in \SymN{N}}{\card{\cycleof{\sigma}{N}}=k}}{
				\sum_{\sumcond{\sigma' \in \SymN{N-k}}{\sigma' = f(\sigma \setminus \cycleof{\sigma}{N})}}{\boltzmann{\Eperm{\rep{\sigma'}}}}}.
	\\
	\intertext{Changing summation order,}
	&=
			\sum_{\sigma' \in \SymN{N-k}}{
				\sum_{\sumcond{\sigma \in \SymN{N}}{f(\sigma \setminus \cycleof{\sigma}{N}) = \sigma'}}{\boltzmann{\Eperm{\rep{\sigma'}}}}}.
	\\
	\intertext{Noting that the summand depends on $\sigma'$ but not on $\sigma$,}
	&=
			\sum_{\sigma' \in \SymN{N-k}}{
				{\boltzmann{\Eperm{\rep{\sigma'}}}} \cdot \card{\set{\sigma \in \SymN{N} \ | \ f(\sigma \setminus \cycleof{\sigma}{N}) = \sigma'}}}. \label{eq:step-num-squeezed}
\end{align}
Consider the number from the set cardinality in~\Cref{eq:step-num-squeezed}. We now show that it is $\binom{N-1}{k-1} \cdot \pfact{k-1}$ for every $\sigma'$. To see this, consider $h_c(\sigma) = f(\sigma \setminus \cycleof{\sigma}{N})$ that is defined the set of permutations $\sigma \in \SymN{N}$ s.t.\ that $\cycleof{\sigma}{N}$ is determined to a specific set $c$. We argue that $h$ is one-to-one from this set to $\SymN{N-k}$.
To see that $h_c$ is surjective, given $\sigma' \in \SymN{N-k}$, we can construct $\sigma$ s.t.\ $\cycleof{\sigma}{N}=c$ and $h_c(\sigma)=\sigma'$ by adding to $c$ the cycles of $\sigma'$ under the ``inverse'' renaming, from the elements $[1,N-k]$ to the elements $[1,N] \setminus c$ while respecting the order.
To see that $h_c$ is injective, given $\sigma_1,\sigma_2 \in \SymN{N}$ with $\cycleof{\sigma_1}{N} = \cycleof{\sigma_2}{N} = c$, if $h_c(\sigma_1) = h_c(\sigma_2)$, then, because $\sigma_1 \setminus \cycleof{\sigma_1}{N}$ and $\sigma_2 \setminus \cycleof{\sigma_2}{N}$ are both permutations over the same set of elements, than the renaming that $f$ applies to both is the same, which means that $\sigma_1 \setminus \cycleof{\sigma_1}{N} = \sigma_2 \setminus \cycleof{\sigma_2}{N}$ and thus also $\sigma_1 = \sigma_2$.
What the fact that $h_c$ is one-to-one on $\SymN{N-k}$ means is that every $\sigma' \in \SymN{N-k}$ has exactly one inverse in $\SymN{N}$ for every choice of the cycle of $N$. There are $\binom{N-1}{k-1} \cdot \pfact{k-1}$ such choices, which is what we wanted.

Plugging this result to~\Cref{eq:step-num-squeezed} yields
\begin{align}
	\sum_{\sumcond{\sigma \in \SymN{N}}{\card{\cycleof{\sigma}{N}}=k}}{\boltzmann{\Eperm{\rep{\sigma \setminus \cycleof{\sigma}{N}}}}}
	&=
	\label{eq:first-cycle-argument-end}
			\sum_{\sigma' \in \SymN{N-k}}{
				\binom{N-1}{k-1} \cdot \pfact{k-1} \cdot {\boltzmann{\Eperm{\rep{\sigma'}}}}}.
	\\
	\intertext{We can now introduce the potential over fewer elements:}
	&=
			\binom{N-1}{k-1} \cdot \pfact{k-1} \cdot \pfact{N-k} \cdot \boltzmann{\Vto{N-k}}.
\end{align}
We now plug this expression back into~\Cref{eq:proof-sum-all-step}:
\begin{align}
\boltzmann{\Vto{N}} &= 
\frac{1}{\fact{N}} \sum_{k=1}^{N}{\boltzmann{\Enk{N}{k}} \cdot \binom{N-1}{k-1} \cdot \pfact{k-1} \cdot \pfact{N-k} \cdot \boltzmann{\Vto{N-k}}}.
\end{align}

The claim follows because
\begin{equation*}
	\frac{1}{\fact{N}} \cdot \binom{N-1}{k-1} \cdot \pfact{k-1} \cdot \pfact{N-k} = \frac{1}{N}.
\end{equation*}
\end{proof}

\subsection{Force as Expected Force}
The expression of the potential as a sum over representative permutations yields an expression for the forces derived from the potential.
As in the main text, we define the \emph{connection probability} distribution over permutations by
\begin{equation*}
	\Prerepperm{\sigma} = \frac{\boltzmann{\Eperm{\rep{\sigma}}}}{\sum_{\sigma}{\boltzmann{\Eperm{\rep{\sigma}}}}} = 
    \frac{\boltzmann{\Eperm{\rep{\sigma}}}}{\fact{N} \cdot \boltzmann{\Vall}}.
\end{equation*}
Then the force due to $\Vall$ can be written as the expectation value of the force in representative permutations:
\label{sec:proof-force-as-expected-force}
\begin{lemma}
\label{lem:expected-force}
For every bead $j$ of particle $\ell$,
\begin{equation*}
	\beadforce{\ell}{j}{\Vall} = 
	\sum_{\sigma}{\Prerepperm{\sigma} \cdot \beadforce{\ell}{j}{\Eperm{\rep{\sigma}}}}.
\end{equation*}
\end{lemma}
\begin{proof}
According to~\Cref{def:backward-potential}, $\Vall = -\frac{1}{\beta} \ln{\left(\frac{1}{\fact{N}} \sum_{\sigma \in \Symfromto{1}{N}}{\boltzmann{\Eperm{\rep{\sigma}}}} \right)}$. Taking the derivative gives
\begin{align*}
	\beadforce{\ell}{j}{\Vall} 
	&= 
		\frac{1}{\beta} \cdot \frac{1}{\frac{1}{\fact{N}}\sum_{\sigma \in \Symfromto{1}{N}}{\boltzmann{\Eperm{\rep{\sigma}}}}} \cdot \left(\frac{1}{\fact{N}} \sum_{\sigma \in \Symfromto{1}{N}}{\boltzmann{\Eperm{\rep{\sigma}}} \cdot \left(-\beta \cdot \beadderive{\ell}{j}{\Eperm{\rep{\sigma}}}\right)}\right)
	\\
	&=
		\frac{1}{\beta} \cdot \frac{1}{\boltzmann{\Vall}} \cdot \frac{1}{\fact{N}} \sum_{\sigma \in \Symfromto{1}{N}}{\boltzmann{\Eperm{\rep{\sigma}}} \cdot (-\beta) \cdot \beadderive{\ell}{j}{\Eperm{\rep{\sigma}}}}
	\\
	&=
		\frac{1}{\fact{N} \cdot \boltzmann{\Vall}} \sum_{\sigma \in \Symfromto{1}{N}}{\boltzmann{\Eperm{\rep{\sigma}}} \cdot \beadforce{\ell}{j}{\Eperm{\rep{\sigma}}}}
\end{align*}
which yields the desired expression since
\begin{equation*}
	\Prerepperm{\sigma} 
	=
		\frac{1}{\sum_{\sigma' \in \Symfromto{1}{N}}{\boltzmann{\Eperm{\rep{\sigma'}}}}} \boltzmann{\Eperm{\rep{\sigma}}}	
	= 
		\frac{1}{\fact{N} \cdot \boltzmann{\Vall}} \boltzmann{\Eperm{\rep{\sigma}}},
\end{equation*}
where the last equality uses~\Cref{def:forward-potential}.
\end{proof}

\subsection{Connection Probabilities}
\label{sec:connection-probabilities-proofs}
This section derives the expressions for the connection probabilities and the recurrence on the potentials used to compute these probabilities.

First we define the potentials $\Vfrom{1},\Vfrom{2},\ldots,\Vfrom{N-1},\Vfrom{N}$, which will allow us to summary the potential on subsets of particles.
First, we define them in a way that is different from the definition in the main text; we derive the latter, which is an equivalent form that is amenable to efficient computation, is reconstructed in~\Cref{thm:backward-recurrence}.
The definition that we use here is similar to the recursion relation that holds for $\Vto{1},\Vto{2},\ldots,\Vto{N-1},\Vto{N}$ per~\Cref{thm:forward-potential-recurrence}, except that here, the recursion is terminated earlier:
\begin{definition}
\label{def:backward-potential}
The potential $\Vfrom{\particledown}$ for $\particledown=1,\ldots,N$ is defined by the recurrence
\begin{align*}
    &\boltzmann{\Vfrom{\particledown}} = \frac{1}{N} \sum_{k=1}^{N - \particledown + 1}{\boltzmann{\left(\Vfromto{\particledown}{N-k} + \Enk{N}{k}\right)}}
    \\
    &\boltzmann{\Vfromto{\particledown}{\particledown-1}} = 1.
\end{align*}
\end{definition}
Notice that for $\particledown=1$ this agrees with the definition of $\Vfromto{1}{N}$, per~\Cref{thm:forward-potential-recurrence}.

Intuitively, $\Vfrom{\particledown}$ is the potential of the ring polymers on particles $[\particledown,N]$ in the representative permutations, which is why it is useful to derive the probabilities that a cycle ends at a particular particle.
Although the recurrence in this definition takes $\bigO(N^2)$ time to evaluate for each $s$, resulting in $\bigO(N^3)$ overall, we will show that this can be reduced to $\bigO(N^2)$ using a different recurrence relation for $\Vfrom{\particledown}$. For now, this definition is useful to derive expressions for the connection probabilities.

\begin{theorem}
\label{lem:direct-link-probability}
For every $1 \leq \ell < N$, %
\begin{equation*}
	\Prrepnext{\ell}{\ell+1} = 1 - \frac{1}{\boltzmann{\Vall}} \boltzmann{\left(\Vto{\ell} + \Vfrom{\ell + 1}\right)}.
\end{equation*}
\end{theorem}
\begin{proof}
Given a permutation $\sigma \in \SymN{N}$, we say that \emph{$\rep{\sigma}$ snips at $\ell+1$} for $\ell+1 \leq N$ if $\rep{\sigma}$ has a cycle $[\ell+1,v]$ for some $v \geq \ell+1$. (Intuitively, as $\repsym$ ``scans'' the interval $[1,N]$ according to the cycles of $\sigma$, it chooses to cut a cycle between $\ell$ and $\ell+1$.)

Observe, from the definition of $\repsym$, that $\nextof{\ell}{\ell+1}{\rep{\sigma}}$ if $\rep{\sigma}$ contains a cycle $[s_1,s_2]$ s.t.\ $s_1 \leq \ell, \ell+1 \leq s_2$, and only if this is the case, for otherwise $\ell,\ell+1$ are in separate cycles and are not connected.
This means that $\nextof{\ell}{\ell+1}{\rep{\sigma}}$ iff $\rep{\sigma}$ does \emph{not} snip at $\ell+1$.
Hence
\begin{align*}
	\Prrepnext{\ell}{\ell+1} 
	&= 1 - \sum_{\sumcond{\sigma \in \SymN{N}}{\snip{\sigma}{\ell+1}}}{\Prerepperm{\sigma}}
	\\
	&= 1 - \frac{1}{\boltzmann{\Vall}}\left(
		\frac{1}{\fact{N}} \sum_{\sumcond{\sigma \in \SymN{N}}{\snip{\sigma}{\ell+1}}}{\boltzmann{\Eperm{\rep{\sigma}}}}	
	\right).
\end{align*}

It thus remains to prove that for every $\ell$
\begin{equation*}
	\frac{1}{\fact{N}} \sum_{\sumcond{\sigma \in \SymN{N}}{\snip{\sigma}{\ell+1}}}{\boltzmann{\Eperm{\rep{\sigma}}}} = 
	\boltzmann{\left(\Vto{\ell} + \Vfrom{\ell + 1}\right)},
\end{equation*} 
which we do by induction on $N \geq \ell$. 
For the base case, when $N = \ell$, this is simply the fact that
\begin{equation*}
	\frac{1}{\fact{N}} \sum_{\sigma \in \SymN{N}}{\boltzmann{\Eperm{\rep{\sigma}}}}	= \boltzmann{\left(\Vto{N} + \Vfrom{N+1}\right)}
	= 
	\boltzmann{\Vto{N}}
\end{equation*}
seeing that $\Vfrom{N+1} = 0$, and from the definition of $\Vto{N}$ (\Cref{def:forward-potential}).

For the induction step, assume that the claim holds for all $\particleup < N$, and prove for $N$.
By the same reasoning as in~\Crefrange{eq:first-cycle-argument-start}{eq:first-cycle-argument-end}---considering the first cycle that $G$ creates and representing the rest of the construction of $\rep{\sigma}$ as the application of $G$ to permutations on a subset of the elements---we have
\begin{align}
	&\frac{1}{\fact{N}} \sum_{\sumcond{\sigma \in \SymN{N}}{\snip{\sigma}{\ell+1}}}{\boltzmann{\Eperm{\rep{\sigma}}}}
	\\
	&\quad =
	\frac{1}{\fact{N}} \sum_{k=1}^{N-\ell}{\binom{N-1}{k-1} \cdot \pfact{k-1} \cdot \boltzmann{\Enk{N}{k}}
		\sum_{\sumcond{\sigma \in \SymN{N-k}}{\snip{\sigma}{\ell+1}}}{
			\boltzmann{\Eperm{\rep{\sigma}}}
		}
	}
	.
	\\
	\intertext{By the induction hypothesis (for $N - k$),}
	&\quad =
	\frac{1}{\fact{N}} \sum_{k=1}^{N-\ell}{\binom{N-1}{k-1} \cdot \pfact{k-1} \cdot \boltzmann{\Enk{N}{k}}
		\cdot \pfact{N-k} \cdot \boltzmann{\left(\Vto{\ell} + \Vfromto{\ell+1}{N-k}\right)}
	}
	\\
	&\quad =
	\boltzmann{\Vto{\ell}} \cdot \frac{1}{N} \sum_{k=1}^{N-\ell}{
		\boltzmann{\left(\Vfromto{\ell+1}{N-k} + \Enk{N}{k}\right)}
	},
	\\
	\intertext{which, due to the recursive definition of $\Vfrom{\particledown}$ (\Cref{def:backward-potential}, with $\particledown = \ell + 1$), is}
	&\quad =
	\boltzmann{\Vto{\ell}} \boltzmann{\Vfromto{\ell+1}{N}},
\end{align}
as required.
\end{proof}

\begin{theorem}
\label{lem:close-cycle-probability}
For every $1 \leq \particledown \leq \ell \leq N$,
\begin{equation*}
	\Prrepnext{\ell}{\particledown} = \frac{1}{\ell} \frac{1}{\boltzmann{\Vall}} {\boltzmann{\left(\Vto{\particledown-1} + \Efromto{\particledown}{\ell} + \Vfrom{\ell+1}\right)}}.
\end{equation*}
\end{theorem}
\begin{proof}
Similarly to before, given a permutation $\sigma \in \SymN{N}$, we say that \emph{$\rep{\sigma}$ snips at $\ell+1$} for $\ell+1 \leq N$ if $\rep{\sigma}$ contains a cycle $[\ell+1,v]$ for some $v \geq \ell+1$.
We say that \emph{$\rep{\sigma}$ snips from $\particledown$ to $\ell$} for $\ell \geq \particledown$ if $\rep{\sigma}$ contains the cycle $[\particledown,\ell]$.

Observe, from the definition of $\repsym$, that $\nextof{\ell}{\particledown}{\rep{\sigma}}$ for $\particledown \leq \ell$ iff $\rep{\sigma}$ contains the cycle $[\particledown,\ell]$; intuitively, under $\repsym$, a particle can be connected to a lower particle only when a cycle closes.
Hence
\begin{align*}
	\Prrepnext{\ell}{\particledown} 
	&= \sum_{\sumcond{\sigma \in \SymN{N}}{\sumcond{\snip{\sigma}{\ell+1}}{\fullsnip{\sigma}{\particledown}{\ell}}}}{\Prerepperm{\sigma}}
	\\
	&= \frac{1}{\boltzmann{\Vall}}\left(
		\frac{1}{\fact{N}} \sum_{\sumcond{\sigma \in \SymN{N}}{\sumcond{\snip{\sigma}{\ell+1}}{\fullsnip{\sigma}{\particledown}{\ell}}}}{\boltzmann{\Eperm{\rep{\sigma}}}}	
	\right).
\end{align*}
It thus remains to show that
\begin{equation*}
	\frac{1}{\fact{N}} \sum_{\sumcond{\sigma \in \SymN{N}}{\sumcond{\snip{\sigma}{\ell+1}}{\fullsnip{\sigma}{\particledown}{\ell}}}}{\boltzmann{\Eperm{\rep{\sigma}}}} = 
	\frac{1}{\ell} \frac{1}{\boltzmann{\Vall}} {\boltzmann{\left(\Vto{\particledown-1} + \Efromto{\particledown}{\ell} + \Vfrom{\ell+1}\right)}},
\end{equation*}
which we do by induction on $N \geq \ell$.
Essentially, The proof uses the same inductive proof as in the previous theorems to ``peel away'' the part $[\ell+1,N]$ through any number of cycles, then cuts the cycle $[\particledown,\ell]$, and continues to consider all permutations on the elements $[1,\particledown-1]$.

For the base case $N = \ell$, the condition that $\rep{\sigma}$ snip at $\ell+1=N+1$ is considered vacuous, and we are at exactly the same scenario we had in the proof of~\Cref{thm:forward-potential-recurrence}, where we are interested in the sum of all permutations where the cycle of $N$ has $k$ elements, for $k = N - \particledown+1$. Exactly as shown there, this is
\begin{align*}
	\frac{1}{\fact{N}} \sum_{\sumcond{\sigma \in \SymN{N}}{\fullsnip{\sigma}{\particledown}{N}}}{\boltzmann{\Eperm{\rep{\sigma}}}}
	&\quad =
	\frac{1}{N} \boltzmann{\left(\Vto{\particledown-1} + \Efromto{\particledown}{N}\right)},
\end{align*}
which is the desired $\frac{1}{\ell} \boltzmann{\left(\Vto{\particledown-1} + \Efromto{\particledown}{\ell}\right)}$ since $N = \ell$.

The induction steps follows the same pattern: summing over all the ways to choose $k$ elements in the cycle containing $N$, and summing over all the permutations on $N-k$ elements on which $\repsym$ snips at the appropriate places:
\begin{equation*}
\begin{split}
	&\frac{1}{\fact{N}} \sum_{\sumcond{\sigma \in \SymN{N}}{\sumcond{\snip{\sigma}{\ell+1}}{\fullsnip{\sigma}{\particledown}{\ell}}}}{\boltzmann{\Eperm{\rep{\sigma}}}}
	\\
	&\quad =
	\frac{1}{\fact{N}} \sum_{k=1}^{N-\ell}{\binom{N-1}{k-1} \cdot \pfact{k-1} \cdot \boltzmann{\Enk{N}{k}}
		\sum_{\sumcond{\sigma \in \SymN{N-k}}{\sumcond{\snip{\sigma}{\ell+1}}{\fullsnip{\sigma}{\particledown}{\ell}}}}{
			\boltzmann{\Eperm{\rep{\sigma}}}
		}
	}
	.
	\\
	\intertext{By the induction hypothesis (for $N - k$),}
	&\quad =
	\frac{1}{\fact{N}} \sum_{k=1}^{N-\ell}{\left[
		\begin{split}
			&\binom{N-1}{k-1} \cdot \pfact{k-1} \cdot \boltzmann{\Enk{N}{k}}
			\\
			&\cdot \pfact{N-k} \cdot \frac{1}{\ell} \cdot \boltzmann{\left(\Vto{\particledown-1} + \Efromto{\particledown}{\ell} + \Vfromto{\ell+1}{N-k}\right)}
		\end{split}
		\right]
		}
	\\
	&\quad =
	\frac{1}{\ell} \boltzmann{\left(\Vto{\particledown-1} + \Efromto{\particledown}{\ell}\right)} \cdot \frac{1}{N} \sum_{k=1}^{N-\ell}{
		\boltzmann{\left(\Vfromto{\ell+1}{N-k} + \Enk{N}{k}\right)}
	},
	\\
	\intertext{which, due to the recursive definition of $\Vfrom{\ell}$,}
	&\quad =
	\frac{1}{\ell} \boltzmann{\left(\Vto{\particledown-1} + \Efromto{\particledown}{\ell}\right)} \boltzmann{\Vfromto{\ell}{N}},
\end{split}
\end{equation*}
and the claim follows.
\end{proof}

\begin{lemma}
\label{thm:zero-connection-probability}
For $\particledown > \ell + 1$,
\begin{equation*}
	\Prrepnext{\ell}{\particledown} = 0.
\end{equation*}
\end{lemma}
\begin{proof}
Easy to see from the definition of $\repsym$ that all the cycles that $\rep{\sigma}$ generates connect $\ell$ to either $\ell+1$ (if they are part of the same cycle) or a to a lower element (if $\ell$ is the endpoint of a cycle).
\end{proof}

Finally, we use the above results about connection probabilities to derive a recurrence relation that enables an efficient computation of $\Vfrom{\particledown}$:
\begin{theorem}
\label{thm:backward-recurrence}
The quantities $\Vfrom{1},\Vfrom{2},\ldots,\Vfrom{N}$ (\Cref{def:backward-potential}) satisfy the recursion formula
\begin{equation*}
\label{thm:tail-recursion}
	\boltzmann{\Vfrom{\particledown}} = \sum_{\ell=\particledown}^{N}{\frac{1}{\ell} \boltzmann{\left(\Efromto{\particledown}{\ell} + \Vfrom{\ell+1}\right)}}
\end{equation*}
with $\Vfrom{N+1} = 0$.
\end{theorem}
\begin{proof}
Every element $\particledown$ is the successor of another element in every $\rep{\sigma}$, so with probability $1$,
\begin{equation*}
	\sum_{\ell=1}^{N}{\Prrepnext{\ell}{\particledown}} = 1.
\end{equation*}

Consider first the case that $\particledown > 1$.
Then
$
\Prrepnext{\ell}{\ell+1} + \sum_{\ell=\particledown}^{N}{\Prrepnext{\ell}{\particledown}} = 1.
$
Using~\Cref{lem:direct-link-probability,lem:close-cycle-probability,thm:zero-connection-probability}, this means that
\begin{align*}
	\quad &\left(1 - \frac{1}{\boltzmann{\Vall}} \boltzmann{\left(\Vto{\particledown-1} + \Vfrom{\particledown}\right)}\right)
	\\
	+ &\sum_{\ell=\particledown}^{N}{\frac{1}{\ell} \frac{1}{\boltzmann{\Vall}} \boltzmann{\left(\Vto{\particledown-1} + \Efromto{\particledown}{\ell} + \Vfrom{\ell+1}\right)}}
	= 1.
\end{align*}
Rearranging,
\begin{align*}
	\sum_{\ell=\particledown}^{N}{\frac{1}{\ell} \boltzmann{\left(\Vto{\particledown-1} + \Efromto{\particledown}{\ell} + \Vfrom{\ell+1}\right)}}
	=
	\boltzmann{\left(\Vto{\particledown-1} + \Vfrom{\particledown})\right)}.
\end{align*}
Isolating $\Vfrom{\particledown}$,
\begin{align*}
	\boltzmann{\Vfrom{\particledown}} = \sum_{\ell=\particledown}^{N}{\frac{1}{\ell} \boltzmann{\left(\Efromto{\particledown}{\ell} + \Vfrom{\ell+1}\right)}},
\end{align*}
as desired.

For the case $\particledown=1$, we have
$
	\sum_{\ell=1}^{N}{\Prrepnext{\ell}{\particledown}} = 1,
$
meaning
\begin{align*}
	&\sum_{\ell=1}^{N}{\frac{1}{\ell} \frac{1}{\boltzmann{\Vall}} \boltzmann{\left(\Efromto{1}{\ell} + \Vfrom{\ell+1}\right)}}
	= 1.
\end{align*}
Moving $\Vfrom{1}$ to the other side,
\begin{equation*}
	\boltzmann{\Vfrom{1}} = \sum_{\ell=1}^{N}{\frac{1}{\ell} \boltzmann{\left(\Efromto{1}{\ell} + \Vfrom{\ell+1}\right)}},
\end{equation*}
as desired.
\end{proof}  \fi

\end{document}